\documentclass[a4paper,11pt]{article}
\usepackage[T2A]{fontenc}
\usepackage[utf8]{inputenc}
\usepackage[english]{babel}
\usepackage{amsmath}
\usepackage{amsfonts}
\usepackage{amssymb}
\usepackage{wasysym}
\usepackage{amsthm}
\usepackage{tikz}
\usepackage{tikz-cd}
\usetikzlibrary{decorations.markings}
\usepackage{ytableau}
\usepackage{longtable}
\usepackage{empheq}
\usepackage[hidelinks,colorlinks=true,unicode]{hyperref}
\hypersetup{
	linkcolor=blue,
	citecolor=blue,
	filecolor=magenta,
	urlcolor=blue,
}
\usepackage{url}
\usepackage[left=2cm,right=2cm,top=2cm,bottom=2cm,bindingoffset=0cm]{geometry}

\newtheorem{st}{Proposition}

\newtheorem{defin}{Definition}

\newtheorem{lemma}{Lemma}
\newtheorem{rem}{Remark}
\newtheorem{nota}{Notation}

\title{{\bf Multiplicity-free $U_q(sl_N)$ 6-j symbols: \\ 
	relations, asymptotics, symmetries} \vspace{.5cm}}
\author{{\bf Victor Alekseev$^{a,b,c}$\thanks{alekseev.va@phystech.edu},
		Andrey Morozov$^{a,b,c}$\thanks{Andrey.Morozov@itep.ru},
		Alexey Sleptsov$^{a,b,c}$\thanks{sleptsov@itep.ru}}\date{ }}

\begin{document}
\maketitle

	\vspace{-7cm}
\begin{center}
	\hfill	ITEP-TH-25/19\\
	\hfill	IITP-TH-17/19\\
	\hfill	MIPT-TH-15/19
\end{center}
\vspace{5.5cm}

\vspace{-1cm}
\begin{center}
	$^a$ {\small {\it Institute for Theoretical and Experimental Physics, Moscow 117218, Russia}}\\
	$^b$ {\small {\it Institute for Information Transmission Problems, Moscow 127994, Russia}}\\
	$^c$ {\small {\it Moscow Institute of Physics and Technology, Dolgoprudny 141701, Russia}}\\
\end{center}
\vspace{0.1cm}

\begin{abstract}
A closed form expression for multiplicity-free  quantum 6-j symbols (MFS) was proposed in \cite{MFS} for symmetric representations of $U_q(sl_N)$, which are the simplest class of multiplicity-free representations. In this paper we rewrite this expression in terms of q-hypergeometric series ${}_4\Phi_3$. We claim that it is possible to express any MFS through the 6-j symbol for  $U_q(sl_2)$ with a certain factor. It gives us a universal tool for the extension of various properties of the quantum 6-j symbols for $U_q(sl_2)$ to the MFS. We demonstrate this idea by deriving the asymptotics of the MFS in terms of associated tetrahedron for classical algebra $U(sl_N)$.
	
Next we study MFS symmetries using known hypergeometric identities such as argument permutations and Sears' transformation. We describe symmetry groups of MFS. As a result we get new symmetries, which are a generalization of the tetrahedral symmetries and the Regge symmetries for $N=2$.
\end{abstract}


\tableofcontents

\section{Introduction}

Racah-Wigner coefficients or 6-j symbols play an important role in mathematics and theoretical physics, because they appear in many different problems. From mathematical point of view they describe the associativity data, which are still unknown for $U_q(sl_N)$. The main difficulty is in the appearance of the so-called multiplicities, which happens when the algebra rank $N$ is greater than 2. However, even for multiplicity-free representations analytical formulas for 6j-symbols are known only for a small class of representations, namely, for symmetric representations. 

In theoretical physics the algebra $U_q(sl_N)$ is very important especially in quantum physics. Here is an incomplete list of topics, in which 6-j symbols of quantum Lie algebra $U_q(sl_N)$ or its classical version $U(sl_N)$, appear:\\
$\bullet$ quantum mechanics \cite{LL} and quantum computing \cite{qcomp},\\
$\bullet$ quantum $\mathcal{R}$-matrices and integrable systems \cite{Rint},\\
$\bullet$ WZW conformal field theory and 3d Chern-Simons theory \cite{WZW1, WZW2},\\
$\bullet$  lattice gauge theory \cite{lattice},\\
$\bullet$ 3-d quantum gravity \cite{qgrav},\\
$\bullet$ quantum $sl_N$ invariants of knots \cite{RT},\\
$\bullet$ Turaev-Viro invariants of 3-manifolds and topological field theory \cite{TV1,TV2},\\
$\bullet$ Drinfeld associator and Kontsevich integral \cite{DA, KI},\\
$\bullet$ orthogonal polynomials \cite{ortpol, ortpol2, ortpol3}.

One can see that 6-j symbols are widely used in both classical and modern works. Note that in many situations, e.g. in the quantum gravity or in statistical models, one considers partition functions, which contain a sum over all possible 6-j symbols of the given gauge group. In such problems it would be very useful to use symmetries between different 6-j symbols in order to reduce the sum and simplify the computation. 


Quantum 6-j symbols have a lot of symmetries, most of them are still unknown.  Nowadays we have different situations for $U_q(sl_2)$ and more general $U_q(sl_N)$ 6-j symbols. All symmetries of $U_q(sl_2)$ 6-j symbols are well known and well studied, many interesting and surprising results are obtained, see e.g. \cite{Roberts, Boalch, brehamet2015regge, sleptsov_new_sym}. In the present paper we are interested in the so-called \textit{linear} symmetries. \textit{Non-linear} symmetries (e.g. the pentagon relation), that are more complicated, are out of the scope of this paper. Linear symmetries of $U_q(sl_2)$ Racah coefficients include Regge symmetries, the tetrahedral symmetries and transformation $q \leftrightarrow q^{-1}$ \cite{klimyk}.  {\it Known} symmetries of $U_q(sl_N)$ include complex conjugation, a $q\leftrightarrow {q^{-1}}$ and the tetrahedral symmetries \cite{WZW2}. 


Some symmetries may be obtained with the help of the eigenvalue hypothesis \cite{NewSymsFromEvHyp, Mironov:2016, cabling, Dhara:2017ukv, Alekseev:2019} including some generalization for Regge symmetries. It says that the Racah matrices are uniquely defined by the eigenvalues of the $\hat{\mathcal{R}}$-matrices. All studied examples says that it is true and this hypothesis becomes a useful tool to derive symmetries. Moreover, there is an exact expression for the Racah matrices through the $\hat{\mathcal{R}}$-matrix eigenvalues for the matrices of the size up to $5\times 5$ \cite{Ev_Hyp} and $6\times 6$ \cite{Universality}.

The 6-j symbols calculation is a big problem for $U_q(sl_N)$ representations. There are few calculation methods and each of them is extremely tedious. Unlike the $U_q(sl_2)$ case, where the answer is known in a closed form for each representation \cite{KR}, the analytical expression for arbitrary representations is still unknown. However, for the special case of  symmetric and conjugated symmetric $U_q(sl_N)$ representations, the analytical expression was proposed recently \cite{MFS, Mironov:2014}. The result gives us plenty of new questions. In particular, which properties of the expression are special for $U_q(sl_2)$ and which can be generalized to the more complex cases. For instance, in this context it was found \cite{racah_pol} that 6-j symbols for symmetric representations of $U_q(sl_N)$ can be expressed in terms of orthogonal q-Racah polynomials as well as their counterpart for $U_q(sl_2)$. Also note that 6j-symbols of $U_q(sl_N)$ for non-symmetric representations were studied in \cite{Morozov:2019haw,Morozov:2019jqp,Morozov:2019kgx}. 

\bigskip

In this paper we study the analytical expression from \cite{MFS} in order to find new symmetries. In section \ref{S2} we start by introducing Racah coefficients and 6-j symbols for $U_q(sl_N)$. In this paper we consider 6-j symbols that have only symmetric and conjugate to symmetric representations. All these 6-j symbols may be transformed via tetrahedral symmetries  into either type I and type II \cite{WZW2}. For type I the only conjugate to symmetric representation is the second one, for type II -- the third one. Each type can be considered as a natural generalization of $U_q(sl_2)$ 6-j symbols because each tensor product decomposition for this case has no multiplicities and can be enumerated by an integer number rather than a whole Young diagram. We consider the expression for both types as an analytic function and study its special properties to obtain new symmetries. 

In section \ref{S3} we simplify the expression. Firstly, we prove that the expression may be reduced and the series became much more similar to $U_q(sl_2)$ series. This was done for both types independently and as it appears they can be represented as one universal expression for both types. Then we express it in terms of q-hypergeometric function $_4\Phi_3$ with some factor. Also it is proven that this expression does not have any inequality restrictions on its arguments, as it was proposed in the original article. As a result, the expression becomes more convenient for studying symmetries.

In section \ref{S4} we analyze the hypergeometric expression of multiplicity-free 6-j symbol. We find the transformation between the multiplicity-free $U_q(sl_N)$ 6-j symbol and its $U_q(sl_2)$ counterpart. This result creates a lot of possibilities to generalize well-known $U_q(sl_2)$ 6-j symbol properties to the considered case. As an immediate output of such relation in section \ref{S5} we derive the classical ($q=1$) 6-j symbol asymptotics, using known results for $U(sl_2)$. Originally it was written in terms of the associated tetrahedron \cite{Ponzano_Regge, Roberts}. The $U(sl_N)$ generalization modifies the expression so that the tetrahedron now depends on $N$ and deforms differently for two types of 6-j symbols.

In section \ref{S6} the resulting 6-j symbol expression has been studied for symmetries. Obtained $_4\Phi_3$ series has two known symmetries: permutations of arguments in each row and the Sears' transformation \cite{gaspar}. The total number of hypergeometric symmetries is 23040 for both types, it was obtained by manual computations on computer. However, only 24 form symmetry group of 6-j symbols for type I and 12 for type II. Some of them are tetrahedral, others can be described as the Regge symmetry generalization for $N\geq 2$. 

We also consider additional symmetries that equates $U_q(sl_N)$ and $U_q(sl_M)$ 6-j symbols in subsections \ref{SS4},\ref{SS5}. Being obtained as symmetries between hypergeometric series, they require a normalizing factor in terms of 6-j symbols. Non-trivial expressions are found for both types and examples are provided. The main results of these subsections are symmetries that generalize permutation in a different from tetrahedral way. They become usual well-known symmetries when $N=2$, but for $N>2$ they depend on $N$ explicitly.

\section{Racah coefficients, 6-j symbols and types I, II expression}\label{S2}

To define 6-j symbols we need firstly to remind the Racah matrix definition. Here we work with q-deformed algebra $U_q(sl_N)$.
Let us consider 3 irreducible $\mathbb{C}$-modules of representations $R_1,R_2,R_3$ acting in $V_{R_1},V_{R_2},V_{R_3}$. Due to a tensor product associativity, $(V_{R_1} \otimes V_{R_2}) \otimes V_{R_3} = V_{R_1} \otimes (V_{R_2} \otimes V_{R_3})$, hence there is a unitary transformation
\begin{align}
U:\ \ (R_1 \otimes R_2) \otimes R_3 &\rightarrow R_1 \otimes (R_2 \otimes R_3).
\end{align}
On the other hand, we can rewrite it in irreducible components, where $M_{X}^{R_1,R_2}$ is a multiplicity space of all $X$'s in the decomposition $R_1\otimes R_2$:
\begin{equation}
\begin{split}
(R_1 \otimes R_2) \otimes R_3 &= \left(\bigoplus_i M_{X_i}^{R_1,R_2} \otimes X_i\right)\otimes R_3 = \bigoplus_{i,k} M_{X_i}^{R_1,R_2} \otimes M_{R_{4_k}}^{X_i,R_3} \otimes R_{4_k},\\
R_1 \otimes (R_2 \otimes R_3) &= R_1 \otimes \left(\bigoplus_j M_{Y_j}^{R_2,R_3} \otimes Y_j\right) = \bigoplus_{j,k} M_{R_{4_k}}^{R_1,Y_j} \otimes M_{Y_j}^{R_2,R_3} \otimes R_{4_k}.
\end{split}
\end{equation}

If we consider some particular $R_4$ in the decomposition, it corresponds to the vector space of representations. A basis constructed from the highest weights' vectors differs for these two fusions.

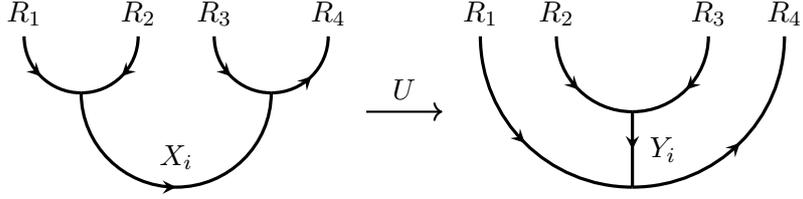
\tikzset{->-/.style={decoration={
			markings,
			mark=at position .5 with {\arrow{stealth}}},postaction={decorate}}}
\begin{center}
\begin{tikzpicture}[scale=1]
\draw [->-, very thick](1.5, 0) arc (0:-90:0.75);
\draw [->-, very thick](0, 0) arc (-180:-90:0.75);
\draw [->-, very thick](3.25, -0.75) arc (-90:0:0.75);
\draw [->-, very thick](2.5, 0) arc (-180:-90:0.75);
\draw [->-, very thick](0.75, -0.75) arc (-180:0:1.25);
\draw node at (0, 0.3) {$R_1$};
\draw node at (1.5, 0.3) {$R_2$};
\draw node at (2.5, 0.3) {$R_3$};
\draw node at (4, 0.3) {$R_4$};
\draw node at (2, -1.6) {$X_i$};
\draw[->, thick] (4.5,-1) to (5.5,-1);
\draw node at (5, -0.7) {$U$};
\draw [->-, very thick](9, 0) arc (0:-90:1);
\draw [->-, very thick](7, 0) arc (-180:-90:1);
\draw [->-, very thick](8, -2) arc (-90:0:2);
\draw [->-, very thick](6, 0) arc (-180:-90:2);
\draw [->-, very thick] (8,-1) to (8,-2);
\draw node at (6, 0.3) {$R_1$};
\draw node at (7, 0.3) {$R_2$};
\draw node at (9, 0.3) {$R_3$};
\draw node at (10, 0.3) {$R_4$};
\draw node at (8.4, -1.5) {$Y_i$};
\end{tikzpicture}
\end{center}
Thus, there is a transformation between two vector spaces that is defined by the  Racah matrix or Racah-Wigner 6-j symbols.
\begin{defin}
	Racah coefficients are elements of Racah matrix that is the map:
	\begin{align}\label{U_mat_def}
	U \left( \begin{matrix}
	R_1 & R_2 \\
	R_3 & {R_4}
	\end{matrix} \right): \bigoplus_{i} M_{X_i}^{R_1,R_2} \otimes M_{R_{4}}^{X_i,R_3} \rightarrow \bigoplus_{j} M_{R_{4}}^{R_1,Y_j} \otimes M_{Y_j}^{R_2,R_3}.
	\end{align}
\end{defin}
\begin{defin}
	Wigner 6-j symbol is the element of a normalized Racah matrix:
	\begin{align}
	\left\lbrace \begin{matrix}\label{6j_def}
	R_1 & R_2 & X_i\\
	R_3 & R_4 & Y_j
	\end{matrix} \right\rbrace =  \frac{1}{\sqrt{\dim_q(X_i)\dim_q(Y_j)}} U_{i,j} \left( \begin{matrix}
	R_1 & R_2 \\
	R_3 & R_4
	\end{matrix} \right).
	\end{align}
\end{defin}
Here $\dim_q$ means the quantum deformation of the usual expression for the dimension of the representation \cite{3SB}. It can be computed for every $U_q(sl_N)$ representation $R$ using the corresponding Young diagram $\lambda$ ($\lambda^T$ is a transposed Young diagram):
\begin{equation}
\dim_q(\lambda)=  \prod_{(i,j)\in \lambda} \frac{q^{\frac{1}{2}(N+i-j)} - q^{-\frac{1} {2} (N+i-j)}} {q^{\frac{1}{2}(\lambda_i-i+\lambda_j^T-j+1)} - q^{-\frac{1}{2}(\lambda_i-i+\lambda_j^T-j+1)}}.
\end{equation}

In this paper we work with the special class of 6-j symbols, which can be seen as a natural generalization of $U_q(sl_2)$ case for $U_q(sl_N)$ 6-j symbols. The initial representations and the resulting one are either symmetric or conjugated to symmetric for this class. Further we will assume that $R_1,R_2,R_3,R_4$ representations are symmetric. Corresponding Young diagrams are $[r_1],[r_2],[r_3],[r_4]$, here $r_n$ are integers that denote numbers of boxes for $U_q(sl_N)$ symmetric representations. Conjugated Young diagram is written as $\overline{[r_n]}$ and correspond to $\overline{R}_n$.
\begin{defin}
	We shall call two 6-j symbols below type {\rm I} and type {\rm II},  $\ytableausetup{boxsize=0.6em, aligntableaux=top}
	\ytableaushort{\cdot}$ means $N-1$ vertical boxes.
	\begin{align}
	\text{I type: }\left\lbrace \begin{matrix}
	[r_1] & \overline{[r_2]} & X\\
	[r_3] & [r_4] & Y
	\end{matrix} \right\rbrace &\equiv \left\lbrace \begin{matrix}
	\ydiagram{2}\dots \ydiagram{1} & \ytableaushort{\cdot}\dots \ytableaushort{\cdot} & \ytableaushort{\cdot\cdot}\dots \ytableaushort{\cdot\none}*{2}\dots \ydiagram{1}\\
	\ydiagram{2}\dots \ydiagram{1} & \ydiagram{2}\dots \ydiagram{1} & \ytableaushort{\cdot\cdot}\dots \ytableaushort {\cdot\none}*{2} \dots \ydiagram{1}
	\end{matrix} \right\rbrace,\\
	\text{II type: }
	\left\lbrace \begin{matrix}
	[r_1] & [r_2] & X\\
	\overline{[r_3]} & [r_4] & Y
	\end{matrix} \right\rbrace &\equiv \left\lbrace \begin{matrix}
	\ydiagram{2}\dots \ydiagram{1} & \ydiagram{2}\dots \ydiagram{1} & \ydiagram{2,2}\dots \ydiagram{2,1}\dots \ydiagram{1}\\
	\ytableaushort{\cdot}\dots \ytableaushort{\cdot}& \ydiagram{2}\dots \ydiagram{1} & \ytableaushort{\cdot\cdot}\dots \ytableaushort{\cdot\none}*{2}\dots \ydiagram{1}
	\end{matrix} \right\rbrace.
	\end{align}
\end{defin}

Although arguments $R_1,R_2,R_3,R_4$ are very simple and can be parametrized by the width and $N$, the last pair of $X$ and $Y$ Young diagrams has more sophisticated expressions. There are two possible cases of tensor products: $[r_n]\otimes [r_m]$ and $[r_n]\otimes \overline{[r_m]}$. Each element in the decomposition depends on the initial pair of representations and the ordering number in the sum. From the Littlewood-Richardson rules \cite{harris} it is easy to see that the mentioned tensor products are multiplicity-free and all representations in a decomposition have different width. Similarly to $U_q(sl_2)$ case, where it is possible to enumerate diagrams by the only integer parameter $i$, for mentioned $U_q(sl_N)$ decompositions we have the enumerating parameter -- the first row length. To shorten the notation we shall write 6-j symbol of type I and type II in a more compact form. Let us denote the type by variable $T \in \{1,2\}$. Type I 6-j symbol is:
\begin{equation}\left[ \begin{matrix}\label{MFS_nota_t1}
r_1 & r_2 & i\\
r_3 & r_4 & j
\end{matrix} \right]_{1}^N := \left\lbrace \begin{matrix}
[r_1] & \overline{[r_2]} & \left[i, \dfrac{r_2-r_1+i}{2}^{N-2} \right]\\
[r_3] & [r_4] & \left[ j, \dfrac{r_2-r_3+j}{2}^{N-2} \right]
\end{matrix} \right\rbrace,
\end{equation}
 and type II:
\begin{equation}\left[ \begin{matrix}\label{MFS_nota_t2}
r_1 & r_2 & i\\
r_3 & r_4 & j
\end{matrix} \right]_{2}^N := \left\lbrace \begin{matrix}
[r_1] & [r_2] & \left[\dfrac{r_1+r_2+i}{2}, \dfrac{r_1+r_2-i}{2} \right]\\
\overline{[r_3]} & [r_4] & \left[ j, \dfrac{r_2-r_3+j}{2}^{N-2} \right]
\end{matrix} \right\rbrace,\end{equation}
where $i,j$ is defined in such a way in order to have a nice $N=2$ limit.

Let us note that the fusion rules restrictions require additional equalities:
\begin{align}\label{fusion_rules}
\begin{split}
r_1+r_3&=r_2+r_4 \hspace{5mm}\text{ for type I,}\\
r_1+r_2&=r_3+r_4 \hspace{5mm}\text{ for type II.}
\end{split}
\end{align}
\begin{defin}
	The equations (\ref{Regge}) between $U_q(sl_2)$ 6-j symbols are called Regge symmetries or Regge transformations \cite{Regge} $(\rho = \frac{r_1 + r_2 + r_3 + r_4}{2}, \rho' = \frac{r_1 + r_3 + i + j}{2}, \rho'' = \frac{r_2 + r_4 + i + j }{2})$:
	\begin{gather}
	\left\lbrace \begin{matrix}\label{Regge}
	r_1 & r_2 & i \\
	r_3 & r_4 & j\end{matrix} \right\rbrace = \left\lbrace \begin{matrix}
	\rho - r_3 & \rho - r_4 & i \\
	\rho - r_1 & \rho - r_2 & j \\
	\end{matrix} \right\rbrace = \left\lbrace \begin{matrix}
	\rho' - r_3 & r_2 & \rho' - j \\
	\rho' - r_1 & r_4 & \rho' - i \\
	\end{matrix} \right\rbrace = \left\lbrace \begin{matrix}
	r_1 & \rho'' - r_4 & \rho'' - j\\
	r_3 & \rho'' - r_2 & \rho'' - i \\
	\end{matrix} \right\rbrace = \\ = \left\lbrace \begin{matrix}
	\rho-r_3 & \rho' - r_4 & \rho'' - j\\
	\rho-r_1 & \rho' - r_2 & \rho'' - i
	\end{matrix} \right\rbrace = \left\lbrace \begin{matrix}
	\rho''-r_3 & \rho-r_4 & \rho' - j\\
	\rho''-r_1 & \rho-r_2 & \rho' - i
	\end{matrix} \right\rbrace \nonumber.
	\end{gather}
\end{defin}
\begin{defin}
	The tetrahedral symmetry is a known property of 6-j symbol to be invariant under row and column permutations \cite{WZW2} ($\lambda_i,\mu,\nu$ are arbitrary Young diagrams):
	\begin{align}\label{tetra}
	\left\{ \begin{matrix}
	\lambda_1 & \lambda_2 & \mu \\
	\lambda_3 & \lambda_4 &\nu \end{matrix} \right\}
	&= \left\{ \begin{matrix}
	\overline{\lambda_3} & \overline{\lambda_2} &\overline{\nu} \\
	\overline{\lambda_1} & \overline{\lambda_4} &\overline{\mu}\end{matrix} \right\}
	= \left\{ \begin{matrix}
	\lambda_3 & \overline{\lambda_4} & \overline{\mu} \\
	\lambda_1 & \overline{\lambda_2} & \overline{\nu}\end{matrix} \right\}    =\\
	&= \left\{ \begin{matrix}
	{\lambda_1} & \overline{\mu} & \overline{\lambda_2} \\
	\overline{\lambda}_3 & \overline{\nu} & \overline{\lambda_4}\end{matrix} \right\}
	=\left\{ \begin{matrix}
	{\lambda_2} & {\lambda_1} & {\mu} \\
	\overline{\lambda}_4 & \overline{\lambda_3} & \overline{\nu}\end{matrix} \right\}\nonumber.
	\end{align}
\end{defin}
\begin{st}
	6-j symbol in $U_q(sl_N), N>2$ with symmetric and conjugate to symmetric representations is either trivial ($X$ and $Y$ has the only possible value) or may be equated by tetrahedral symmetry and conjugation to either type {\rm I} or type {\rm II}.
\end{st}
\begin{proof}
	There are only a few possible variants to write down a 6-j symbol with symmetric and conjugate to symmetric representations. By conjugation of 6-j symbol we can transform $R_4$ to a symmetric diagram. Thus, let us prove the proposition without loss of generality only for symmetric $R_4$. Let us now investigate how the first three arguments may be organized. There are four different cases that correspond to the number of conjugated representations in the product.
\begin{itemize}
	\item All three representations are conjugated.

Let us conjugate all terms in the product $\overline{[r_1]} \otimes \overline{[r_2]} \otimes \overline{[r_3]} \supset [r_4]$, so we can consider ${[r_1]} \otimes {[r_2]} \otimes {[r_3]} \supset \overline{[r_4]}$ and $N>2$. It is obvious from the fusion rules \cite{harris} that for $N>4$ it is not possible to combine the representations into a conjugated one because there are no more than 3 rows in a resulting Young diagram, whereas $\overline{[r_4]}$ has $N-1>3$ rows.

Now we need to prove that it is not possible even for $N={3,4}$. The $N=4$ case requires the rows of $R_4$ to be equal. The Littlewood-Richardson rules \cite{harris} say that the resulting diagram is constructed as the first multiplier with the second multiplier's elements but with some restrictions. For symmetric diagrams they forbid to put the new elements in one column. Hence, if we need to combine diagrams into a rectangular one, the corresponding 6-j symbol is trivial. Indeed, the only way to combine the diagrams properly is to consider them equal and to put them under each other.

Here and below we use some non-negative integer parameters $a,b,c$ that encode a Young diagram, the aim of these parameters is to specify the shape of a considered diagram.

The $N=3$ case has a $\overline{[r_4]}$ diagram that may be written as $[a,a]$. The $[a,a]$ is trivial, because there is the only diagram $X = [r_1+r_2-b,b]$ that has width $a$. Indeed, if the width is smaller, the third multiplier can not make the second row width equal to $a$, if it is greater, we can not make $R_4$ anymore.

Therefore, all $N>2$ 6-j symbols with 3 conjugated representations are trivial.

	\item All three representations are symmetric.

Obviously, if $R_1, R_2, R_3, R_4$ are symmetric in $U_q(sl_N), N>3$, then the corresponding 6-j symbol has the only $X = [r_1+r_2]$, the same for $Y$. If $N=3$, there is a possibility to make a Young diagram with columns of height $N$. However, the fusion rules restrict $X = [r_1+r_2-a, a] = [b+r_4,b]$, hence $X = [r_1+r_2+r_4, r_1+r_2-r_4]$ and this 6-j symbol is trivial.

	\item Two representations are conjugated and one is symmetric.

Note, that the multiplicity of $R_4$ in decomposition $R_1\otimes R_2\otimes R_3$ does not change under a permutation of multipliers. Hence we may always decompose the product of conjugated representations and then multiply it by the symmetric one. Without loss of generality we consider $\overline{[r_1]} \otimes \overline{[r_2]} \otimes {[r_3]}$.

Let us firstly decompose the product of conjugated representations. In general, it has the diagram $[a^{N-2},b]$, where $b\le a$. It is obtained from $[(r_1+r_2)^{N-2},r_1+r_2-c,c]$ by reducing the column of height $N$. If $N>3$, the product $[a^{N-2},b] \otimes [r_3]$ may have a symmetric diagram in the decomposition only if $a=b$, but it will be trivial because $X = [(r_3-r_4)^{N-1}]$. If $N=3$, $[a,b]\otimes [r_3]$ easily makes symmetric diagram with condition $X=[a,a+r_3-r_4]$. But we can find $a$ from the $\overline{[r_1]} \otimes \overline{[r_2]}$ decomposition and it is unique for fixed $r_1$ and $r_2$.

As a result, there are no non-trivial 6-j symbols with two conjugated symmetric representations and symmetric $R_4$.

 \item One conjugated representation.

There are three such 6-j symbols up to a conjugation:
	\begin{align}
	\left\lbrace \begin{matrix}
	\overline{[r_1]} & [r_2] & X\\
	[r_3] & [r_4] & Y
	\end{matrix} \right\rbrace,
	\left\lbrace \begin{matrix}
	[r_1] & \overline{[r_2]} & X\\
	[r_3] & [r_4] & Y
	\end{matrix} \right\rbrace,
	\left\lbrace \begin{matrix}
	[r_1] & [r_2] & X\\
	\overline{[r_3]} & [r_4] & Y
	\end{matrix} \right\rbrace.
	\end{align}
One can check that they may be nontrivial.
\end{itemize}
We can apply a tetrahedral symmetry to these 6-j symbols, in particular, row permutation of arguments $(R_1, R_2) \leftrightarrow (R_3,  R_4)$. After this transformation the first and the third 6-j symbols are swapped and the second one is invariant. Applying other symmetries, one can check that type I and type II are not equated by tetrahedral symmetries.
\end{proof}

It is worth mentioning that there are tetrahedral symmetries acting within each type. In particular, a type I 6-j symbol is still type I after row permutations and the swap of the first two columns. Type II is conserved only by the row permutation of the first two columns. These are the only tetrahedral symmetries that possible to derive if one consider symmetries of type I or type II. The others either were used earlier to transform 6-j symbol into one of the types, or transform any type into a completely different 6-j symbol, which has non-symmetric representations and much more complicated structure, so they are out of the scope of the present paper.

\bigskip

The expression for 6-j symbol of type I and II was proposed in \cite{MFS}. It may be written as follows.

{\small\begin{gather}\label{MFS}
	\left[ \begin{matrix}
	r_1 & r_2 & i\\
	r_3 & r_4 & j
	\end{matrix} \right]_{T}^N
	=\theta_N\left(r_1,r_2,i\right) \theta_N\left(r_3, r_4, i\right) \theta_N\left(r_1, r_4, j\right) \theta_N\left(r_2, r_3, j\right) [N-1]_q![N-2]_q!\sum_{z = z_{min}}^{z_{max}}\\\nonumber
	\dfrac{  (-1)^z [z+N-1]_q!\cdot A_{T,z} }{[ z-\frac{r_1+r_2+i}{2}]_q![z-\frac{r_3+r_4+i}{2}]_q! [z-\frac{r_1+r_4+j}{2}]_q! [z-\frac{r_2+r_3-j}{2}]_q! [\frac{r_1+r_2+r_3+r_4}{2}-z]_q! [\frac{i+j+r_1+r_3}{2}-z]_q! [\frac{i+j+r_2+r_4}{2}-z]_q!}, \\
	\theta_N(a,b,c)=\sqrt{\dfrac{[\frac{a+b-c}{2}]_q! [\frac{c+a-b}{2}]_q! [\frac{b+c-a}{2}]_q!}{[\frac{a+b+c}{2}+N-1]_q!}}, \qquad A_{T,z} = \left[\begin{matrix}
	\dfrac{[k+z_{min}-z]_q!}{[k+z_{min}+N-2 - z]_q!} \hspace{5mm}\text{ for type I,}\\
	\dfrac{[k-z_{max}+ z]_q!}{[k-z_{max} +N-2+ z]_q!} \hspace{5mm}\text{ for type II.}
	\end{matrix}\right.
	\end{gather}}
To write the 6-j symbol expression we use quantum numbers notations. It is by the definition $[n]_q= \frac{q^{\frac{n}{2}} -q^{-\frac{n}{2}}}{q^{\frac{1}{2}} -q^{-\frac{1}{2}}}$. Quantum generalization of factorials for non-negative integers can be written as $[n]_q! = \prod_{k=1}^{n}[k]_q$. Also $k = \frac{1}{2}\min(i-r_1+r_2, j-r_3+r_2)$ and $z_{min},z_{max}$ are defined as the smallest and the largest integers for which the summand is non-trivial , i.e. there are no factorials of negative integers. The expression differs for two types  only in the $A_{T,z}$ expression. Also the following conditions were imposed in the original paper \cite{MFS} (as we show below, they are not necessary):
\begin{equation}
\label{conds}
\left\{\begin{matrix}
0 \le r_2\le r_1 \le r_3 \\ 0 \le r_1\le r_2
\end{matrix}\right. \hspace{5mm} \left.\begin{matrix}
\text{for type I,}\\ \text{for type II.}
\end{matrix}\right.
\end{equation}

\section{Hypergeometric expression for 6-j symbols}\label{S3}

In this section we express the 6-j symbol expression in terms of basic q-hypergeometric series $_4\Phi_3$. Firstly, we define the q-hypergeometric functions and remind their symmetric properties. After this we use the inequality properties (\ref{conds}) to simplify the 6-j symbol expression. We prove with the help of tetrahedral symmetries that the 6-j symbol's domain may be extended beyond the mentioned inequalities. Then we write the obtained series as a $_4\Phi_3$ function. As a result, both types can be written as q-hypergeometric $_4\Phi_3$ series multiplied by some factor.

\subsection{q-Hypergeometric symmetries}
A q-Pochhammer symbol is defined as $(a,q)_n = \prod_{k=0}^{n-1}(1-a q^k)$.

\begin{defin}
	The q-hypergeometric series are defined as:
	\begin{align}
	_{p+1}\phi_p\left(
	\begin{matrix}
	a_1, \ldots, a_{p+1}\\
	b_1, \ldots, b_p
	\end{matrix};q,z
	\right) := \sum_{n=0}^{\infty} \dfrac{(a_1,q)_n\ldots (a_{p+1},q)_n}{(b_1,q)_n\ldots (b_p,q)_n (q,q)_n}z^n.
	\end{align}
	It can be also rewritten in a form, which is more convenient for us:
	\begin{align}\label{hyp_def}
	{}_{p+1}\Phi_p\left(
	\begin{matrix}
	a_1, \ldots, a_p, a_{p+1} \\
	b_1, \ldots, b_p
	\end{matrix};q,z
	\right) := {}_{p+1}\phi_p\left(
	\begin{matrix}
	q^{a_1}, \ldots, q^{a_p}, q^{a_{p+1}} \\
	q^{b_1}, \ldots, q^{b_p}
	\end{matrix};q,z
	\right).
	\end{align}
\end{defin}

It is far more convenient because it may be reformulated in terms of q-factorials:

\begin{align}
{}_{p+1}\Phi_p\left(
\begin{matrix}
a_1+1, \ldots, a_p+1, a_{p+1}+1 \\
b_1+1, \ldots, b_p+1
\end{matrix};q,z
\right) = \sum_{n=0}^{\infty} \dfrac{[a_1+n]_q!}{[a_1]_q!}\ldots \dfrac{[a_{p+1}+n]_q!}{[a_{p+1}]_q!} \dfrac{[b_1]_q!}{[b_1+n]_q!}\ldots \dfrac{[b_p]_q!}{[b_p+n]_q!} \dfrac{z^n}{[n]_q!}.
\end{align}
This expression evidently has the limit $\lim\limits_{q\rightarrow 1}[a]_q!= a!$, where the whole series becomes a usual hypergeometric function.

There are a lot of known symmetries for $_4\Phi_3$ series. Here we consider only permutation symmetry and Sears' transformation.
\begin{defin}
	Permutation symmetry is the evident property of ${}_r\Phi_p$ functions to be invariant under permutations $\omega \in \mathbb{S}_r$ and $u\in \mathbb{S}_p$:
	\begin{equation}\label{trans_hyp_perm}
	{}_r\Phi_p\left( \begin{matrix}
	a_1,\ldots,a_r\\
	b_1,\ldots,b_p
	\end{matrix} ;q,z\right)=
	{}_r\Phi_p\left( \begin{matrix}
	a_{\omega(1)},\ldots,a_{\omega(r)}\\
	b_{u(1)},\ldots,b_{u(p)}
	\end{matrix} ; q,z \right).
	\end{equation}
\end{defin}
\begin{defin}
	Sears' transformation \cite{gaspar} is the relation between two ${}_4\Phi_3$ functions:
	\begin{equation}\label{trans_sears}
	\begin{split}
	{}_4\Phi_3\left( \begin{matrix}
	x,y,z,n\\
	u,v,w
	\end{matrix} ;q,q\right)=
	\dfrac{[v{-}z{-}n{-}1]_q![u{-}z{-}n{-}1]_q![v{-}1]_q![u{-}1]_q!}{[v{-}z{-}1]_q![v{-}n{-}1]_q![u{-}z{-}1]_q![u{-}n{-}1]_q!}\ {}_4\Phi_3\left( \begin{matrix}
	w-x,w-y,z,n\\
	1{-}u{+}z{+}n,1{-}v{+}z{+}n,w
	\end{matrix} ; q,q \right),
	\end{split}
	\end{equation}
	where $x+y+z+n+1=u+v+w$.
\end{defin}

\subsection{6-j symbol as \texorpdfstring{$_5\Phi_4$}\ \ series}
Let us denote the sum (\ref{MFS}) as $\left[ \begin{matrix}
r_1 & r_2 & i\\
r_3 & r_4 & j
\end{matrix} \right]_{T}^N  = K'\cdot\sum_m I_m = K' \cdot I$, where $m = \frac{1}{2}(r_1+r_2+ r_3+r_4)-z$. Then it can be easily rewritten as:
{\small\begin{gather}
I = \sum_{m = m_{min}}^{m_{max}}\dfrac{ (-1)^{\frac{r_1+r_2+r_3+r_4}{2}-m} [\frac{r_1+r_2+r_3+r_4}{2}-m+N-1]_q! \cdot A_{T,m}} {[m]_q![\frac{r_3+r_4-i}{2}-m]_q! [\frac{r_1+r_2-i}{2}-m]_q! [\frac{r_2+r_3-j}{2}-m]_q! [\frac{r_1+r_4-j}{2}-m]_q! [\frac{i+j-r_2-r_4}{2}+m]_q! [\frac{i+j-r_1-r_3}{2}+m]_q! },\\
K'=\theta_N\left(r_1,r_2,i\right) \theta_N\left(r_3, r_4, i\right) \theta_N\left(r_1, r_4, j\right) \theta_N\left(r_2, r_3, j\right) [N-1]_q![N-2]_q!\ ,\\
A_{T,m} = \left[\begin{matrix}
\dfrac{[k-m_{max}+m]_q!}{[k-m_{max}+N-2 + m]_q!} \hspace{5mm}\text{ for type I,}\\
\dfrac{[k+m_{min}- m]_q!}{[k+m_{min} +N-2- m]_q!} \hspace{5mm}\text{ for type II.}
\end{matrix}\right.
\end{gather}}
The explicit relations for $m_{min}$ and $m_{max}$ can be easily found from the denominator factorials, because the summand is zero if and only if there is a negative factorial in the denominator:
\begin{equation}
m_{max} = \frac{1}{2}\min\begin{pmatrix}
r_1+r_2-i\\ r_3+r_4-i\\ r_1+r_4-j\\ r_2+r_3-j
\end{pmatrix}, \hspace{10mm} m_{min} = \frac{1}{2}\max\begin{pmatrix}
0\\ r_1+r_3-i-j\\ r_2+r_4-i-j
\end{pmatrix}.
\end{equation}
As it can be derived from fusion rules, $k,m_{max},m_{min}$ are always integers when a 6-j symbol exists. Moreover, $k$ has a clear meaning in terms of Young diagrams -- it is the minimum width among the conjugated parts of diagrams, corresponding to $X_i$ and $Y_j$.

One can notice, that the considered expression fits the $_5\Phi_4$ definition (\ref{hyp_def}), if $z=q$. This allows us to claim the following.
\begin{st}
	Both type {\rm I} and type {\rm II} may be written as $_5\Phi_4$ q-hypergeometric series multiplied by simple factors:
	{\small\begin{gather}\left[ \begin{matrix}
	r_1 & r_2 & i\\
	r_3 & r_4 & j
	\end{matrix} \right]_T^N = K'' \cdot  {}_5\Phi_4 \left( \begin{matrix}
	a_1,a_2,a_3,a_4,a_5\\
	 b_1,b_2,b_3,b_4
	\end{matrix}; q,q \right),\\
	2a_i = \left( \begin{matrix}
	2\{k-m_{max}+1 , -k - m_{min}-N+2\}_T\\
	-r_1-r_2+i \\
	 -r_3-r_4+i\\
	 -r_1-r_4+j\\
	  -r_2-r_3+j
	\end{matrix} \right),
	\qquad 2b_i = \left( \begin{matrix} 
	 -r_1-r_2-r_3-r_4-2(N-1)\\
	 i+j-r_2-r_4 + 2\\
	   i+j-r_1-r_3 + 2\\
	   2\{k-m_{max}+N-1 , -k - m_{min}\}_T \\
	\end{matrix} \right),\\
	\begin{gathered}
	K'' =  \dfrac{ K' \cdot A_{T,0} \cdot [\frac{r_1+r_2+r_3+r_4}{2}+N-1]_q! } {[\frac{r_3+r_4-i}{2}]_q! [\frac{r_1+r_2-i}{2}]_q! [\frac{r_2+r_3-j}{2}]_q! [\frac{r_1+r_4-j}{2}]_q! [\frac{i+j-r_2-r_4}{2}]_q! [\frac{i+j-r_1-r_3}{2}]_q! },
	\end{gathered}
	\end{gather}}
where $\{e_1,e_2\}_T \equiv e_T$ is $e_1$ for type {\rm I} and $e_2$ for type {\rm II}.
\end{st}
 It can be proven straightforwardly by substitution of q-Pochhammer symbols.

\subsection{Expression of 6-j symbol as \texorpdfstring{${}_4\Phi_3$}\ \ series}

The obtained expression for 6-j symbol is not quite convenient to find its symmetries. Expressions for $k$, $m_{min}$ and $m_{max}$ may be simplified in the following way.
\begin{lemma}\label{L1}
	For all type {\rm I} 6-j symbols $k-m_{max} =\frac{i+j-r_1-r_3}{2}$ if the following conditions are satisfied:
	\begin{equation}\label{type1_cond}
	\begin{cases}
	r_2\le r_1\le r_3,\\ 
	r_1+r_3=r_2+r_4.
	\end{cases}
	\end{equation}
\end{lemma}\label{L2}
\begin{proof}
	Let us consider $k-m_{max}=\frac{i+j-r_1-r_3}{2}$. One can check that there are 2 cases when it is so, hence they may be written as the union of two systems:
	\begin{equation}
	\begin{cases}
	r_1+r_2-i\le r_3+r_4-i,\\
	r_1+r_2-i\le r_2+r_3-j,\\
	r_1+r_2-i\le r_1+r_4-j,\\
	j-r_3\le i-r_1;
	\end{cases} \hspace{10mm} \begin{cases}
	r_2+r_3-j\le r_3+r_4-i,\\
	r_2+r_3-j\le r_1+r_2-i,\\
	r_2+r_3-j\le r_1+r_4-j,\\
	i-r_1\le j-r_3.
	\end{cases}
	\end{equation}
	If the conditions (\ref{type1_cond}) satisfied, the first three inequalities are true. The union of these two systems may be reduced to the next expression.
	\begin{equation}
	\left[\begin{matrix}
	j-i\le r_4-r_2, \\
	j-i\ge r_4-r_2.
	\end{matrix}\right.
	\end{equation}
	Consequently, every 6-j symbol from type I is described by $k-m_{max} = \frac{i+j-r_1-r_3}{2}$.
\end{proof}
\begin{lemma}\label{L3}
	For all type {\rm II} 6-j symbols $k+m_{min} = \frac{r_1+r_2-i}{2}$  if the conditions are satisfied:
	\begin{equation}\label{17}
	\begin{cases}
	r_1\le r_2,\\ 
	r_1+r_2=r_3+r_4.
	\end{cases}
	\end{equation}
\end{lemma}
\begin{proof}
	The proof for type II is analogous to type I.
\end{proof}

\begin{lemma}
	Conditions on arguments of a 6-j symbol (\ref{conds}) are redundant, i.e the expression (\ref{MFS}) is valid even if the inequalities are not satisfied.
\end{lemma}
\begin{proof}
	We are able to obtain every possible 6-j symbol of types I and II by performing a tetrahedral symmetry (\ref{tetra}) that leaves the type invariant:
	\begin{gather}
	\left\{ \begin{matrix}
	[r_1] & \overline{[r_2]} & X \\
	[r_3] & [r_4] & Y \end{matrix} \right\}
	= \left\{ \begin{matrix}
	{[r_3]} & \overline{[r_2]} &\overline{Y} \\
	{[r_1]} & {[r_4]} &\overline{X}\end{matrix} \right\}
	= \left\{ \begin{matrix}
	{[r_2]} & \overline{[r_1]} & \overline{X} \\
	[r_4] & {[r_3]} & {Y}\end{matrix} \right\}.
	\end{gather}
	One may immediately notice that these symmetries may transform a 6-j symbol from region $r_2\le r_1 \le r_3$ into all possible representations. The problem is that the expression for the transformed 6-j symbols may differ from the initial expression. We can check it by substituting arguments transformed by tetrahedral symmetries. Let us show that in our notations it acts on $r_1,r_2,r_3,r_4, i, j$ as a permutation. For $R_n$, the symmetry evidently acts as a permutation of $r_n$. There are also representations $X$ and $Y$ that is conjugated, we can consider only diagram $\left[ j, \frac{r_2-r_3+j}{2}^{N-2} \right]$ as an example. Under conjugation it transforms $\left[ j, \frac{r_2-r_3+j}{2}^{N-2} \right] \rightarrow \left[ j, \frac{r_3-r_2+j}{2}^{N-2} \right]$, but the expression depends only on $j$ that is invariant under conjugation.

	Therefore, tetrahedral symmetry acts on the expression as a permutation of arguments. One can check that it is invariant under written tetrahedral symmetry transformation. The same for type II, but we need only one relation (the inequality is $r_1\le r_2$):
	\begin{align}
	\left\{ \begin{matrix}
	[r_1] & {[r_2]} & X \\
	\overline{[r_3]} & [r_4] & Y \end{matrix} \right\}
	= \left\{ \begin{matrix}
	{[r_3]} & {[r_4]} & {X} \\
	\overline{[r_1]} & {[r_2]} & {Y}\end{matrix} \right\}.
	\end{align}
	The symmetry acts non-trivially only on $r_1,r_2,r_3,r_4$, we already showed why it is a permutation. It is easy to see that the expression is invariant under such a transformation.

Therefore, the expression does not change when we write a 6-j symbol without additional inequality conditions (\ref{conds}). Then we can get rid of these conditions as even if they are not satisfied the expression is valid.
\end{proof}

We have proven in Lemma \ref{L1} that for arguments satisfying the inequality condition (\ref{conds}) there are only one combination of $k-m_{max}$ that is present for type I 6-j symbols. This results into the exact value of $A_{T,m}$ which allow us to reduce the whole series. Then we apply tetrahedral symmetries to prove that the statement is true for all type I 6-j symbols. The same procedure has been done for type II and this allows us to simplify both expressions and write down them as follows.

{\small\begin{empheq}[box=\fbox]{gather}
\left[ \begin{matrix}
r_1 & r_2 & i\\
r_3 & r_4 & j
\end{matrix} \right]_T^N = K'\sum_{m = m_{min}}^{m_{max}}\dfrac{ (-1)^{\frac{r_1+r_2+r_3+r_4}{2}-m} [\frac{r_1+r_2+r_3+r_4}{2}+N-1-m]_q!} {[m]_q![\frac{r_3+r_4-i}{2}-m]_q! [\frac{r_2+r_3-j}{2}-m]_q! [\frac{r_1+r_4-j}{2}-m]_q! [\frac{i+j-r_2-r_4}{2}+m]_q! }\times\\\nonumber
\times \dfrac{1}{[\frac{r_1+r_2-i}{2}+ (N-2)\delta_{T,2} -m]_q![\frac{i+j-r_1-r_3}{2}+ (N-2)\delta_{T,1}+m]_q!}.
\end{empheq}}
\bigskip

We can express all factorials as q-Pochhammer symbols. The substitution differs for factorials with $+m$ and $-m$:
\begin{align}\label{poch}
[m_0+m]_q! =  [m_0]_q!(q^{m_0+1},q)_{m}\cdot \dfrac{q^{-\frac{m}{4}(2m_0+m-1)}}{(1-q)^{m}}\ ,
\quad
[m_0-m]_q! = \dfrac{(-1)^{m}[m_0]_q! } {(q^{-m_0},q)_{m}} \cdot \dfrac{(1-q)^{m}} { q^{\frac{m}{4}(2m_0+m-1)}}.
\end{align}

By substituting this to the main expression one can check that among depending on $m$ terms only q-Pochhammer symbols remain. This allows us to write the series as a hypergeometric function:
\begin{align}\left[ \begin{matrix}
r_1 & r_2 & i\\
r_3 & r_4 & j
\end{matrix} \right]_T \sim  {}_4\Phi_3 \left( \begin{matrix}
a_1,a_2,a_3,a_4\\ b_1,b_2,b_3
\end{matrix}; q,q \right).
\end{align}

The ${}_4\Phi_3$ arguments may be easily obtained using (\ref{poch}). Note, that there is the following relation on the arguments:
\begin{align}\label{hyp_relation}
a_1+a_2+a_3+a_4+1=b_1 + b_2 + b_3.
\end{align}
And the factorizable part of the expression:
{\small\begin{equation}\label{coef_K}
K_T = \dfrac{\theta_N\left(r_1,r_2,i\right) \theta_N\left(r_3, r_4, i\right) \theta_N\left(r_1, r_4, j\right) \theta_N\left(r_2, r_3, j\right) [N-1]_q![N-2]_q! [\frac{r_1+r_2+r_3+r_4}{2}+N-1]_q! }{[\frac{r_3+r_4-i}{2}]_q! [\frac{r_1+r_2-i}{2} + (N-2)\delta_{T,2}]_q! [\frac{r_2+r_3-j}{2}]_q! [\frac{r_1+r_4-j}{2}]_q! [\frac{i+j-r_2-r_4}{2}]_q! [\frac{i+j-r_1-r_3}{2}+(N-2)\delta_{T,1}]_q!}.
\end{equation}}
Combing all this into one, we come to the following statement.
\begin{st}
	The considered 6-j symbol expression may me expressed as a $_4\Phi_3$ function for both types. The factor $K_T$ is as in (\ref{coef_K}).
\end{st}

\begin{empheq}[box=\fbox]{gather}
\left[ \begin{matrix}
r_1 & r_2 & i\\
r_3 & r_4 & j
\end{matrix} \right]_T^N
=K_T\cdot {}_4\Phi_3 \left( \begin{matrix}
a_1,a_2,a_3,a_4\\ b_1,b_2,b_3
\end{matrix}; q,q \right)\label{final_expr_hyp},\\\label{final_expr}
2 a_i = \left( \begin{matrix}
-r_1-r_2+i - 2(N-2)\delta_{T,2}\\ -r_3-r_4+i\\ -r_1-r_4+j\\ -r_2-r_3+j
\end{matrix} \right),
\qquad 2 b_i = \left( \begin{matrix}  -r_1-r_2-r_3-r_4-2(N-1)\\ i+j-r_2-r_4+2\\  i+j-r_1-r_3 + 2 + 2(N-2)\delta_{T,1}\\\end{matrix} \right).
\end{empheq}
\bigskip

This is the most suitable form of 6-j symbol for our aims. As it can be seen, we reduced the $_5\Phi_4$ series to the $_4\Phi_3$ one. This is a non-obvious result. In order to proceed with this reduction we used tetrahedral symmetry along with the special properties of the considered two types. Due to the fact that $U_q(sl_2)$ 6-j symbols are expressed via $_4\Phi_3$ too, we may use the same techniques to obtain new results, also the limit $N=2$ is very easy to apply. This result gives us an idea of a strong connection between 6-j symbols and q-hypergeometric series. For example, it is interesting whether all multiplicity-free 6-j symbols can be expressed as $_4\Phi_3$ series.

It is interesting to analyze the number of independent parameters in the obtained expression. Neglecting $q$, on both sides we have 7 parameters: $\{r_1,r_2,r_3,r_4,i,j,N\}$ and  $\{a_1,a_2,a_3,a_4, b_1,b_2,b_3\}$. They are not independent, it was mentioned that, on the one hand, each type has restrictions for $N>2$ that fix one parameter. On the other hand, obtained ${}_4\Phi_3$ series satisfies a balance condition $\sum_{i}a_i+1=\sum_{i}b_i$. Thus, for $N>2$ there are 6 parameters on both sides. For $N=2$, the fusion rules do not fix $r_n$, so there are 6 parameters on both sides. It is natural to ask whether there is a connection between the fusion rules and the balance condition. It seems like these equalities have different meaning, because the condition on $\{a_i,b_i\}$ is satisfied even if $r_1+r_3\neq r_2+r_4$. From this point of view another question arises: what class of 6-j symbols can be described in terms of ${}_4\Phi_3$ series with such equality? This question is out of our consideration in this paper, but it is still important and interesting to study.

\section{Relation with $U_q(sl_2)$ 6-j symbols}\label{S4}

In this section we investigate the relation between 6-j symbols in multiplicity-free $U_q(sl_N)$ and $U_q(sl_2)$  cases. As we have seen, the core of both expressions are $_4\Phi_3$ hypergeometric series. We have already mentioned the number of independent parameters in the series, but now we analyze it in details. Then we shall see the interesting connection between the usual $U_q(sl_2)$ 6-j symbol and considered one.

Let us write down the $_4\Phi_3$ arguments as a vector space with the basis $(r_1,r_2,r_3,r_4,i,j,N)$. We put all the additional constants in $\vec{C}$ since they do not play any role in the next discussion:
\begin{equation}
\begin{pmatrix}
	r_1 + r_2 + r_3 + r_4 + 2(N-1)\\
	r_1 + r_2 - i + 2(N-2)\delta_{T,2}\\
	r_3 + r_4 - i\\
	r_1 + r_4 - j\\
	r_2 + r_3 - j\\
	-r_2 - r_4 + i + j + 2\\
	i + j - r_1 - r_3 + 2(N-1)\delta_{T,1} \\
\end{pmatrix} = \begin{pmatrix}
1 & 1 & 1 & 1 & 0 & 0 & 1\\
1 & 1 & 0 & 0 & -1 & 0 & 2\delta_{T,2}\\
0 & 0 & 1 & 1 & -1 & 0 & 0\\
1 & 0 & 0 & 1 & 0 & -1 & 0\\
0 & 1 & 1 & 0 & 0 & -1 & 0\\
0 & 1 & 0 & 1 & -1 & -1 & 0\\
-1 & 0 & -1 & 0 & 1 & 1 & 2\delta_{T,1}\\
\end{pmatrix}
\begin{pmatrix}
r_1\\
r_2\\
r_3\\
r_4\\
i\\
j\\
N
\end{pmatrix}
+\vec{C}.
\end{equation}

The rank of this matrix is 6, so there is a kernel of dimension one. This kernel is described by a zero vector $\vec{v}$. Note that (\ref{hyp_relation}) is a completely different condition that does not depend on the values of parameters. The zero vector can be written as follows
\begin{align}
\label{zerotrans}
\vec{v}=\begin{cases}
\begin{pmatrix}
0,1,0,1,1,1,-1
\end{pmatrix}, \quad 	\text{Type I,}\\
\begin{pmatrix}
1,1,0,0,0,1,-1
\end{pmatrix}, \quad 	\text{Type II},
\end{cases}
\end{align}
with the corresponding shift in the parameters being
\begin{align}
\alpha\vec{v}=\begin{cases}
	\alpha\vec{v}=\alpha (r_2 + r_4 + i + j -  N), \quad 	\text{Type I,}
	\\
	\alpha\vec{v}=\alpha (r_1 + r_2 + j -  N), \quad 	\text{Type II}.
\end{cases}
\end{align}

This freedom allows to shift the arguments value without changing the actual value of the hypergeometric series, so it can be considered as a symmetry for 6-j symbol although for hypergeometric series it is tautological equality. If one examines the transformation for type I 6-j symbol, it can be seen that the fusion rules are in conflict with it. Indeed, the non-trivial transformation changes $r_2+r_4$, but leaves $r_1+r_3$ unchanged, thus (\ref{fusion_rules}) forbids such transformation for $N>2$, for either type I or type II. However for $N=2$ the fusion rules disappear and we can apply it without any problems. So we take $U_q(sl_N)$ 6-j symbol and make transformation \eqref{zerotrans} in order to get the expression for $U_q(sl_2)$ 6-j symbol:
\begin{equation}\label{Nshift}
\begin{split}
_4\Phi_3(r_1,r_2,r_3,r_4,i,j,{\bf N})_{1} \ &= \ (-1)^N \cdot {}_4\Phi_3(r_1,r_2+N{-}2,r_3,r_4+N{-}2,i+N{-}2,j+N{-}2,{\bf 2}), \\
_4\Phi_3(r_1,r_2,r_3,r_4,i,j,{\bf N})_{2} \ &= \ (-1)^N \cdot {}_4\Phi_3(r_1+N{-}2,r_2+N{-}2,r_3,r_4,i,j+N{-}2,{\bf 2}).
\end{split}
\end{equation}

The only part of expression that differs is the factor $K_T$. It partly replicates the hypergeometric arguments, so only a few terms are left in the relation between of multiplicity free $U_q(sl_N)$ 6-j symbols and $U_q(sl_2)$ ones. For the sake of brevity, we will write the hypergeometric function from (\ref{final_expr_hyp}) as $_4\Phi_3(r_1,r_2,r_3,r_4,i,j,N)_T$. The factor $K'$ changes after transformations, let us write it down explicitly. 

\begin{align}
K'(N)=&\theta_N\left(r_1,r_2,i\right) \theta_N\left(r_3, r_4, i\right) \theta_N\left(r_1, r_4, j\right) \theta_N\left(r_2, r_3, j\right) [N-1]_q![N-2]_q!\ ,\\
\Theta_T(N):=&\dfrac{1}{[N{-}1]_q! [N{-}2]_q!}\dfrac{K'(N)}{K'(2)}.
\end{align}


\begin{align}\label{coef_theta}
\Theta_1(N)&=
\left(\prod_{m=1}^{N-2} \left[\frac{i {-} r_1 {+} r_2}{2}+m\right]_q \left[\frac{j {+} r_2 {-} r_3}{2}+m\right]_q  \left[\frac{j {-} r_1 {+} r_4}{2}+m\right]_q \left[\frac{i {-} r_3 {+} r_4}{2}+m\right]_q \right)^{-\frac{1}{2}},\\
\Theta_2(N)&=
\left(\prod_{m=1}^{N-2} \left[\frac{r_1 {+} r_2{-}i}{2}+m\right]_q \left[\frac{j {+} r_2 {-} r_3}{2}+m\right]_q  \left[\frac{j {+} r_1 {-} r_4}{2}+m\right]_q \left[\frac{i {+} r_3 {+} r_4}{2}{+}1+m\right]_q \right)^{-\frac{1}{2}}.
\end{align}

The resulting relation between multiplicity-free $U_q(sl_N)$ and $U_q(sl_2)$ 6-j symbol is as follows.
\begin{empheq}[box=\fbox]{align}
\left[ \begin{matrix} \label{rollback}
r_1 & r_2 & i\\
r_3 & r_4 & j
\end{matrix} \right]_1^N
&=\left\{\begin{matrix}
r_1& r_2+N-2 &i+N-2 \\
r_3 & r_4+N-2  & j+N-2
\end{matrix} \right\}(-1)^N [N-1]_q! [N-2]_q! \cdot \Theta_1(N),\\
\left[ \begin{matrix}
r_1 & r_2 & i\\
r_3 & r_4 & j
\end{matrix} \right]_2^N
&=\left\{\begin{matrix}
r_1+N-2 & r_2+N-2 &i\\
r_3 & r_4 & j+N-2
\end{matrix} \right\}(-1)^N [N-1]_q! [N-2]_q! \cdot \Theta_2(N).
\nonumber
\end{empheq}

It can be easily checked that the remaining fusion rules for $N=2$ (triangle inequality, etc.) are always satisfied and the resulting 6-j symbol is non-trivial. On the other hand, if one tries to transform $U_q(sl_2)$ 6-j symbol into $N>2$ one, the number of problems arises and it is not possible in general. For example, if $r_1+r_3-r_2-r_4>0$, there is no corresponding $N>2$ 6-j symbol.

This result is interesting not only because it reveals the hidden relation between two classes of 6-j symbols, but additionally it can be applied to extend a lot of known properties of $U_q(sl_2)$ to arbitrary $N$. In the next section we derive the asymptotics formula for the multiplicity-free case. Let us show an example of such a generalization.

\section{Asymptotics of 6-j symbol}\label{S5}
The 6-j symbol asymptotics formula for $N=2, q=1$ was conjectured by G.Ponzano and T.Regge \cite{Ponzano_Regge} and later was proven by J. Roberts \cite{Roberts}. It is formulated in terms of tetrahedron that is combined from the edges of length $J_n:=r_n+1/2,J_5:=i+1/2,J_6:=j+1/2$ and approximates the limit $\lambda\rightarrow \infty $ for representations $\{\lambda r_n, \lambda i, \lambda j\}$: 
\begin{equation}\label{tetr_sl2}
\left\{\begin{matrix}
r_1 & r_2 &i\\
r_3 & r_4 & j
\end{matrix} \right\} \sim \dfrac{1}{\sqrt{12\pi |V(J_n)|}} \cos\left(\sum_{n=1}^{6} J_n \cdot  \Omega(J_n) + \dfrac{\pi}{4} \right),
\end{equation}
where $V$ is the tetrahedron volume, $\Omega_i$ is the external dihedral angle about the edge $J_i$.

Let us consider 6-j symbols at $q=1$. Using \eqref{rollback} we can find the asymptotics for $U(sl_N)$ 6-j symbol as an asymptotics for equal $U(sl_2)$ 6-j symbol. It looks very similar to \eqref{tetr_sl2}, but with deformed expressions for edges, volume and angles. The tetrahedron is now made of $\widetilde{J}_n$ edges, which can be found from $U(sl_N)$ $J_n$:  
\begin{equation}
\begin{cases}
\widetilde{J}_m = J_m,\\
\widetilde{J}_n = J_n+N-2,
\end{cases}
\end{equation}
where $m$ and $n$ are defined differently for two types:
\begin{align}\label{J_def}
 m \in \text{\{1, 3\}}, \quad n \in \text{\{2, 4, 5, 6\}} \qquad \text{Type I},\\
 m \in \text{\{3, 4, 5\} }, \quad  n \in \text{\{1, 2, 6\}} \qquad \text{Type II}. \nonumber
\end{align}
The corresponding volume and angles are denoted by $\widetilde{V}$ and $\widetilde{\Omega}_n$.

The resulting asymptotics for 6-j symbol corresponding to arbitrary symmetric representations of $U_q(sl_N)$, thus, can be written in terms of the associated tetrahedron, but now the tetrahedron depends on $N$:
\begin{empheq}[box=\fbox]{equation}
\dfrac{1}{\Theta_T(N)}\left[ \begin{matrix} \label{asympt}
r_1 & r_2 & i\\
r_3 & r_4 & j
\end{matrix} \right]_T^N
\sim \dfrac{(-1)^N\cdot(N-1)!\cdot(N-2)!}{\sqrt{12\pi\cdot |V(\widetilde{J}_n)|}} \cos\left(\sum_{n=1}^{6} \widetilde{J}_n \cdot \Omega(\widetilde{J}_n) + \dfrac{\pi}{4} \right).
\end{empheq}

Although the factor is quite long for the general case, it becomes much simpler when all $r_n$ coincide, for example, for type I it looks like:

\begin{equation}
\dfrac{\left(\frac{i}{2}+N{-}2\right)!}{\left(\frac{i}{2}\right)!}\dfrac{\left(\frac{j}{2}+N{-}2\right)!}{\left(\frac{j}{2}\right)!}\left[ \begin{matrix}
r & r & i\\
r & r & j
\end{matrix} \right]_{T=1}^N
\sim \dfrac{(-1)^N(N-1)!(N-2)!}{\sqrt{12\pi |V(\widetilde{J}_n)|}}  \cos\left(\sum_{i=1}^{6} \widetilde{J}_n \cdot \Omega(\widetilde{J}_n) + \dfrac{\pi}{4} \right).
\end{equation}

Let us note, that the generalized formula when all parameters of 6-j symbol are the same does not correspond to the regular tetrahedron if $N>2$. Due to this fact we can not simplify the relation further. Interestingly, the resulting tetrahedron is deformed for every type differently. In particular, type II corresponds to the trigonal  pyramid, whereas type I is a bent tetrahedron, which is combined of 4 equal isosceles triangles. 

\section{Symmetries derivation}\label{S6}
\subsection{Hypergeometric symmetries group}
In this subsection we do not write any symmetries explicitly. Here we are describing the structure of obtained symmetries. The statements in this subsection are given without analytical proof, but it has been checked manually.


We use both permutation symmetry (\ref{trans_hyp_perm}) and Sears' transformation (\ref{trans_sears}) in order to get all possible 6-j symbol transformations. The arbitrary composition of Sears' transformations and permutations can be written as:
\begin{align}\label{om_sys}
{}_4\Phi_3 \left( \begin{matrix}
a_1,a_2,a_3,a_4\\ b_1,b_2,b_3
\end{matrix}; q,q \right)= \widetilde{C}\cdot
{}_4\Phi_3 \left( \begin{matrix}
\widetilde{a}_1,\widetilde{a}_2,\widetilde{a}_3,\widetilde{a}_4\\ \widetilde{b}_1,\widetilde{b}_2,\widetilde{b}_3
\end{matrix}; q,q \right),
\end{align}
where variables with $ \ \widetilde{}\ $ denotes the resulting arguments. There is a factor $C$ that appears after Sears' transformations, but we are not interested in it for now. The resulting symmetry has the following form:
\begin{equation}
\left[ \begin{matrix}
r_1 & r_2 & i\\
r_3 & r_4 & j
\end{matrix} \right]_T^N = C\left[ \begin{matrix}
\widetilde{r}_1 & \widetilde{r}_2 & \widetilde{i}\\
\widetilde{r}_3 & \widetilde{r}_4 & \widetilde{j}
\end{matrix} \right]_T^M,
\end{equation}
where $\widetilde{r}_n,\widetilde{i},\widetilde{j}$ are some linear combination of $r_n,i,j$ obtained by the mentioned transformations. Parameters $N,M$ denote the ranks of the corresponding algebras.

To find the symmetries we have to solve the linear system of equations on arguments $\widetilde{r}_n,\widetilde{i},\widetilde{j},M$. Initially we consider $M=N$ to get a unique solution. The rank of the system is 6, because the hypergeometric function has 7 arguments with one additional constraint. Note, that we do not restrict them to the fusion rules when we solve the system. That is done because Sears' transformation do not respect the fusion rules, but some of its combinations with permutations do. Hence we need to obtain all symmetries and then recover fusion rules using (\ref{Nshift}). In this subsection we do not consider the relation (\ref{Nshift}) as a symmetry because it is used to satisfy fusion rules by fixing parameter $M$.  

\begin{st}
	The overall set of symmetries $G$ that contains all compositions of permutations and Sears' transformation is a group and it has 23040 elements in total \cite{23040}.
\end{st}
For $N=M=2$ case this group was discovered in \cite{23040_Doyle}, where it was called 22.5K group. They claimed that it is in fact Coxeter group $D6$, which arises in hyperbolic geometry as the group of hyperbolic tetrahedra symmetries. The volume of a hyperbolic tetrahedron is known to be connected with the quantum 6-j symbol of $U_q(sl_2)$ in an appropriate limit \cite{Murakami_Yano}.

Our result was obtained via the computer algebra system. Permutations and Sears' transformations were programmed explicitly and combined multiple times. By fixing all the constraints on permutations and Sears' transformation, the program reached 23040 elements. It was checked that they are closed under composition. Each symmetry is non-degenerate due to the non-degeneracy of the initial equations, hence all elements are invertible. As a result, 23040 symmetries including identity form a group.

Most of these symmetries cannot be applied in 6-j symbols because they often do not preserve the positiveness of $r_n,i,j$. Also the structure of its subgroups is not clear and it makes the analysis more complicated. Thus, we are interested only in the subgroup that generalizes the known set of symmetries from $N=M=2$ to arbitrary $N$ and $M$, let us call it $S\subset G$. There are 144 elements in $S$ and it is analogous to the $U_q(sl_2)$ group of permutations and Regge transformations, which we denote as $H=S\big|_{N=M=2}\big.$. Moreover, these groups are in one to one correspondence: each symmetry for $N\neq 2 \neq M$ may be transformed to a $N=2=M$ symmetry and vice versa. Note, that the found symmetries from $S$ are well-defined for hypergeometric series, but for 6-j symbols they require the positiveness of $r_n,i,j,M-2$. 

The other symmetries from $G$ are out of our scope in the next discussion. The reformulation of symmetries from $G$ in terms of 6-j symbol have some difficulties. On the one hand, the number of group elements is too large to analyze the symmetries manually, on the other hand the subgroups structure is still unclear. Also there are a lot of symmetries that do not preserve the positiveness of representation parameters, so a lot of symmetries can not be applied to 6-j symbols. Interestingly, the whole group may be obtained as a combination of symmetries $S$ and the following one:
\begin{equation}
\left[ \begin{matrix}
r_1 & r_2 & i\\
r_3 & r_4 & j
\end{matrix} \right]_T^N = \left[ \begin{matrix}
r_1 & r_2 & i\\
-r_3-1 & r_4 & j
\end{matrix} \right]_T^N.
\end{equation}


\bigskip

After the transformation (\ref{Nshift}) is used to find $M$, it is natural to consider two classes of symmetries: one for $N=M$ and another for $N\neq M$.

	\begin{defin}
	If the symmetry requires $N=M$, we call it the internal one, else we call it the external symmetry.  The set of internal  and external symmetries are denoted by $I$ and $E$ respectively.
\end{defin}

Let us provide this definition with examples of both internal and external symmetries.

The internal symmetry:
\begin{equation}
\left[ \begin{matrix}
r_1 & r_2 & i\\
r_3 & r_4 & j
\end{matrix} \right]_2^N = \left[ \begin{matrix}
r_2 & r_1 & i\\
r_4 & r_3 & j
\end{matrix} \right]_2^N.
\end{equation}
The fusion rules (\ref{fusion_rules}) formally require two equalities for LHS and RHS. However, they are linearly dependent in this case, so the equality for one side yields the equality for the other side. Moreover, if $N\neq M\neq 2$, the conditions are in contradiction. 

The external symmetry:
\begin{equation}
\left[ \begin{matrix}
r_1 & r_2 & i\\
r_3 & r_4 & j
\end{matrix} \right]_1^N = C\left[ \begin{matrix}
r_1 & i+N-M & r_2+N-M\\
r_3 & j+N-M & r_4+N-M
\end{matrix} \right]_1^M,
\end{equation}
where $C$ is some factor. Here we have to restrict representations by two equalities: $r_1+r_3= r_2+r_4$ and $r_1+r_3 = i + j+2(N-M)$, so we should fix $2M=2N+i+j-r_1-r_3$:
\begin{equation}
\left[ \begin{matrix}
r_1 & r_2 & i\\
r_3 & r_4 & j
\end{matrix} \right]_1^N = C\left[ \begin{matrix}
r_1 & \dfrac{r_2+r_4+i-j}{2} & \dfrac{3r_2+r_4-i-j}{2}\\
r_3 & \dfrac{r_2+r_4-i+j}{2} & \dfrac{r_2+3r_4-i-j}{2}
\end{matrix} \right]_1^{N+\frac{i+j-r_2-r_4}{2}}.
\end{equation}
Parameters of the transformed 6-j symbol on the RHS have to be non-negative. Parameters $\widetilde{r}_n,\widetilde{i},\widetilde{j}$ are non-negative for each external symmetry, as it will be derived in Appendix.  On the other hand, $M$ still have to be grater then or equal to 2, so not all 6-j symbols may be transformed by this symmetry. Each external symmetry induces a subset of 6-j symbols that has such a relation.

\begin{st}\label{st5}
	For any non-trivial 6-j symbol $\left[ \begin{matrix}
	r_1 & r_2 & i\\
	r_3 & r_4 & j
	\end{matrix} \right]_T^N$ the external symmetry of any type transforms it into the 6-j symbol with non-negative $\widetilde{r}_n,\widetilde{i},\widetilde{j}$.
\end{st}

The proof of this statement uses explicit relations for 6-j symbol symmetries and it is proven in Appendix.

\begin{st}
	The internal symmetries of 6-j symbols form group $I$ with the following structure. It is isomorphic to either $\mathbb{S}_4$ for type {\rm I} or $\mathbb{S}_3 \times \mathbb{Z}_2$ for type {\rm II}. 
\end{st}

If we consider only internal symmetries, we obtain subgroup $I\subset S$. One can check in a straightforward way that $|I|=24$ for type I, $|I|=12$ for type II and the symmetries are isomorphic to mentioned groups. 
\begin{align}
&G \quad \supset \quad S \stackrel{N=M}{\supset} I, \qquad E:=S/I,\nonumber\\\nonumber
\text{Type I:}\qquad&S \cong \mathbb{S}_4 \times \mathbb{S}_3, \qquad I \cong \mathbb{S}_4, \\\nonumber
\text{Type II:}\qquad&S \cong \mathbb{S}_4 \times \mathbb{S}_3, \qquad I \cong \mathbb{S}_3 \times \mathbb{Z}_2,\\\nonumber
&|G|=23040,\ |S|=144
\end{align}

\bigskip

The explicit relations are written in the next subsections. The internal symmetries from $I$ may be applied to any 6-j symbol of the corresponding type. In other words, for every $r_n,i,j$ with the satisfied fusion rules it is possible to write down all symmetries from $I$.

External symmetries relate 6-j symbols for different algebras. There are two important things to note here. Firstly, 6-j symbols and $_4\Phi_3$ differs by a factor that is not always invariant under external symmetries, so we need to add a normalizing factor to this symmetry. Secondly, since there are two group ranks $N$ and $M$, both of them should be greater than or equal to 2 for the symmetry to be valid. As a result, it may be applied only to the part of all type I and type II 6-j symbols.


Let us note that for $U_q(sl_2)$ there are no restrictions from fusion rules, therefore $S$ coincides with $I$ and we have all $144$ symmetries \cite{klimyk}. 


\begin{rem}
		Both internal and external symmetries can be derived using the relation (\ref{rollback}) between $U_q(sl_2)$ 6-j symbols and MFS. 
	\end{rem}
	This method may also be used to check the obtained equalities. If one expresses the list of MFS symmetries as $U_q(sl_2)$ 6-j symbols equalities, factors can be reduced and the equalities form the list of $U_q(sl_2)$ symmetries.

\subsection{Type I internal symmetries}\label{SS2}

In this subsection we write down the internal symmetries of type I.
These symmetries are very similar to the known ones and can be seen as a natural generalization of the symmetries from $U_q(sl_2)$, although in terms of Young diagrams it's not obvious. In the shortened notation it is easy to see the correspondence between $U_q(sl_2)$ and $U_q(sl_N)$ symmetries. Although the internal symmetries of type I by definition need $r_1+r_3 = r_2 + r_4$ to be satisfied, we do not write it explicitly because in every equality either both 6-j symbols exist or both of them do not. The same idea is used for type II internal symmetries. To write down the symmetries in a more compact way, we use the following variables:
\begin{equation}\label{rho_def}
\rho = \dfrac{r_1+r_2+r_3 +r_4}{2} \hspace{5mm} \rho' = \dfrac{r_2+i +r_4+j}{2} = \dfrac{r_1+i+r_3 +j}{2} = \rho''.
\end{equation}

All 6-j symbols below are equal and form group $I$. Columns of the equality list correspond to row permutations, rows correspond to Regge symmetries analogue:
{\small
\begin{align}
\left[ \begin{matrix}
r_1 & r_2 & i\\
r_3 & r_4 & j
\end{matrix} \right]_1^N \hspace{22pt} = \hspace{22pt} \left[ \begin{matrix}
r_3 & r_4 & i\\
r_1 & r_2 & j
\end{matrix} \right]_1^N \hspace{22pt} &= \hspace{22pt} \left[ \begin{matrix}
r_1 & r_4 & j\\
r_3 & r_2 & i
\end{matrix} \right]_1^N \hspace{22pt} = \hspace{22pt}\left[ \begin{matrix}
r_3 & r_2 & j\\
r_1 & r_4 & i
\end{matrix} \right]_1^N \\=\nonumber
\left[ \begin{matrix}
\rho-r_4 & \rho-r_3 & i\\
\rho-r_2 & \rho-r_1 & j
\end{matrix} \right]_1^N \hspace{6pt} = \hspace{6pt} \left[ \begin{matrix}
\rho-r_2 & \rho-r_1 & i\\
\rho-r_4 & \rho-r_3 & j
\end{matrix} \right]_1^N \hspace{6pt} &= \hspace{6pt} \left[ \begin{matrix}
\rho-r_4 & \rho-r_1 & j\\
\rho-r_2 & \rho-r_3 & i
\end{matrix} \right]_1^N \hspace{6pt} = \hspace{6pt} \left[ \begin{matrix}
\rho-r_2 & \rho-r_3 & j\\
\rho-r_4 & \rho-r_1 & i
\end{matrix} \right]_1^N \\=\nonumber
\left[ \begin{matrix}
r_1 & \rho'-j & \rho'-r_4\\
r_3 & \rho'-i & \rho'-r_2
\end{matrix} \right]_1^N \hspace{3pt} = \hspace{3pt} \left[ \begin{matrix}
r_3 & \rho'-i & \rho'-r_4\\
r_1 & \rho'-j & \rho'-r_2
\end{matrix} \right]_1^N \hspace{3pt} &= \hspace{3pt} \left[ \begin{matrix}
r_1 & \rho'-i & \rho'-r_2\\
r_3 & \rho'-j & \rho'-r_4
\end{matrix} \right]_1^N \hspace{3pt} = \hspace{3pt} \left[ \begin{matrix}
r_3 & \rho'-j & \rho'-r_2\\
r_1 & \rho'-i & \rho'-r_4
\end{matrix} \right]_1^N \\= \nonumber
\left[ \begin{matrix}
\rho''-j & r_2 & \rho''-r_3\\
\rho''-i & r_4 & \rho''-r_1
\end{matrix} \right]_1^N \hspace{1pt} = \hspace{1pt} \left[ \begin{matrix}
\rho''-i & r_4 & \rho''-r_3\\
\rho''-j & r_2 & \rho''-r_1
\end{matrix} \right]_1^N \hspace{1pt} &= \hspace{1pt} \left[ \begin{matrix}
\rho''-j & r_4 & \rho''-r_1\\
\rho''-i & r_2 & \rho''-r_3
\end{matrix} \right]_1^N \hspace{1pt} = \hspace{1pt} \left[ \begin{matrix}
\rho''-i & r_2 & \rho''-r_1\\
\rho''-j & r_4 & \rho''-r_3
\end{matrix} \right]_1^N \\= \nonumber
\left[ \begin{matrix}
\rho''{-}j & \rho{-}r_3 & \rho'{-}r_4\\
\rho''{-}i & \rho{-}r_1 & \rho'{-}r_2
\end{matrix} \right]_1^N = \left[ \begin{matrix}
\rho''{-}i & \rho{-}r_1 & \rho'{-}r_4\\
\rho''{-}j & \rho{-}r_3 & \rho'{-}r_2
\end{matrix} \right]_1^N &= \left[ \begin{matrix}
\rho''{-}j & \rho{-}r_1 & \rho'{-}r_2\\
\rho''{-}i & \rho{-}r_3 & \rho'{-}r_4
\end{matrix} \right]_1^N = \left[ \begin{matrix}
\rho''{-}i & \rho{-}r_3 & \rho'{-}r_2\\
\rho''{-}j & \rho{-}r_1 & \rho'{-}r_4
\end{matrix} \right]_1^N \\= \nonumber
\left[ \begin{matrix}
\rho{-}r_4 & \rho'{-}j & \rho''{-}r_3\\
\rho{-}r_2 & \rho'{-}i & \rho''{-}r_1
\end{matrix} \right]_1^N = \left[ \begin{matrix}
\rho{-}r_2 & \rho'{-}i & \rho''{-}r_3\\
\rho{-}r_4 & \rho'{-}j & \rho''{-}r_1
\end{matrix} \right]_1^N &= \left[ \begin{matrix}
\rho{-}r_4 & \rho'{-}i & \rho''{-}r_1\\
\rho{-}r_2 & \rho'{-}j & \rho''{-}r_3
\end{matrix} \right]_1^N = \left[ \begin{matrix}
\rho{-}r_2 & \rho'{-}j & \rho''{-}r_1\\
\rho{-}r_4 & \rho'{-}i & \rho''{-}r_3
\end{matrix} \right]_1^N.
\end{align}
}

These 24 symmetries form a representation of group $I$ mentioned above. It has two notable subgroups: row permutations and Regge transformations analogue. The isomorphism $I \cong \mathbb{S}_4$ is as follows. Permutations from the first row correspond to $\{(),(12)(34), (14)(23), (13)(24)\}$. The first column symmetries correspond to $\{(),(12), (23), (13), (123), (132)\}$. All others can be read from the table:
$$\begin{tabular}{|c|c|c|c|}
\hline 
() & (12)(34) & (14)(23) & (13)(24) \\ 
\hline 
(12) & (34) & (1324) & (1423) \\ 
\hline 
(23) & (1243) & (14) & (1342) \\ 
\hline 
(13) & (1432) & (1234) & (24) \\ 
\hline 
(123) & (243) & (134) & (142) \\ 
\hline 
(132) & (143) & (124) & (234) \\ 
\hline 
\end{tabular} $$

We can write down the generalization of Regge transformations (\ref{Regge}):
\begin{empheq}[box=\fbox]{align}
\left[ \begin{matrix}
r_1 & r_2 & i\\
r_3 & r_4 & j
\end{matrix} \right]_1^N =& \left[ \begin {matrix} \label{regge_I}
{ r_1}&\rho' - j&\rho' - r_4\\ { r_3}&\rho' - i&\rho' - r_2\end {matrix} \right]_1^N=
\left[ \begin {matrix}
 \rho' - j &r_2&\rho' - r_3\\ \rho' - i&r_4&\rho' - r_1\end {matrix} \right]_1^N.
\end{empheq}

Let us give a couple of examples of these symmetries:
\begin{itemize}
	\item Regge symmetry analogue, type I (1\textsuperscript{st} column is invariant, $N\ge2$):
	{\small \begin{align*}
	\hspace{-5mm}\left\lbrace \begin{matrix}
	[8] & \overline{[4]} & [12,4^{N-2}]\\
	[10] & [14] & [6]
	\end{matrix} \right\rbrace&=\left\lbrace \begin{matrix}
	[8] & \overline{[6]} & [14,6^{N-2}]\\
	[10] & [12] & [4]
	\end{matrix} \right\rbrace,\
	\left\lbrace \begin{matrix}
	[10] & \overline{[8]} & [18,8^{N-2}]\\
	[12] & [14] & [6,5^{N-2}]
	\end{matrix} \right\rbrace=\left\lbrace \begin{matrix}
	[10] & \overline{[5]} & [15,5^{N-2}]\\
	[12] & [17] & [9,8^{N-2}]
	\end{matrix} \right\rbrace,\\
	\hspace{-5mm}\left\lbrace \begin{matrix}
	[12] & \overline{[6]} & [16,5^{N-2}]\\
	[14] & [20] & [8]
	\end{matrix} \right\rbrace&=\left\lbrace \begin{matrix}
	[12] & \overline{[9]} & [19,8^{N-2}]\\
	[14] & [17] & [5,5^{N-2}]
	\end{matrix} \right\rbrace,\
	\left\lbrace \begin{matrix}
	[12] & \overline{[8]} & [10,3^{N-2}]\\
	[14] & [18] & [6]
	\end{matrix} \right\rbrace=\left\lbrace \begin{matrix}
	[12] & \overline{[11]} & [13,6^{N-2}]\\
	[14] & [15] & [3]
	\end{matrix} \right\rbrace.
	\end{align*}}
	\item Regge symmetry analogue, type I (2\textsuperscript{nd} column is invariant, $N\ge2$):
	{\small\begin{align*}
	\left\lbrace \begin{matrix}
	[4] & \overline{[6]} & [2,2^{N-2}]\\
	[3] & [1] & [5,4^{N-2}]
	\end{matrix} \right\rbrace&=\left\lbrace \begin{matrix}
	[2] & \overline{[6]} & [4,4^{N-2}]\\
	[5] & [1] & [3,2^{N-2}]
	\end{matrix} \right\rbrace,\
	\left\lbrace \begin{matrix}
	[6] & \overline{[5]} & [7,3^{N-2}]\\
	[3] & [4] & [2,2^{N-2}]
	\end{matrix} \right\rbrace=\left\lbrace \begin{matrix}
	[7] & \overline{[5]} & [6,2^{N-2}]\\
	[2] & [4] & [3,3^{N-2}]
	\end{matrix} \right\rbrace,\\
	\left\lbrace \begin{matrix}
	[5] & \overline{[6]} & [7,4^{N-2}]\\
	[4] & [3] & [8,5^{N-2}]
	\end{matrix} \right\rbrace&=\left\lbrace \begin{matrix}
	[4] & \overline{[6]} & [8,5^{N-2}]\\
	[5] & [3] & [7,4^{N-2}]
	\end{matrix} \right\rbrace,\
	\left\lbrace \begin{matrix}
	[4] & \overline{[6]} & [2,2^{N-2}]\\
	[5] & [3] & [7,4^{N-2}]
	\end{matrix} \right\rbrace=\left\lbrace \begin{matrix}
	[2] & \overline{[6]} & [4,4^{N-2}]\\
	[7] & [3] & [5,2^{N-2}]
	\end{matrix} \right\rbrace.\\
	\end{align*}}
\end{itemize}

\subsection{Type II internal symmetries}\label{SS3}

One can similarly consider type II, there are only 12 symmetries. For brevity we use the following variables:
\begin{equation}
\rho = \dfrac{r_1+r_2+r_3 +r_4}{2} \hspace{5mm} \rho' = \dfrac{r_2+i +r_4+j}{2} \hspace{5mm} \rho'' = \dfrac{r_1+i+r_3 +j}{2}.
\end{equation}
All 6-j symbols below are equal. Columns of the table correspond to a column permutation, rows correspond to Regge symmetries.

\begin{align}
	\left[ \begin{matrix}
	r_1 & r_2 & i\\
	r_3 & r_4 & j
	\end{matrix} \right]_2^N &= \left[ \begin{matrix}
	r_2 & r_1 & i\\
	r_4 & r_3 & j
	\end{matrix} \right]_2^N \\=  \left[ \begin{matrix}\nonumber
	\rho-r_3 & \rho-r_4 & i\\
	\rho-r_1 & \rho-r_2 & j
	\end{matrix} \right]_2^N &= \left[ \begin{matrix}
	\rho-r_4 & \rho-r_3 & i\\
	\rho-r_2 & \rho-r_1 & j
	\end{matrix} \right]_2^N \\= \left[ \begin{matrix}\nonumber
	r_1 & \rho' - r_4 & \rho' - j\\
	r_3 & \rho' - r_2 & \rho' - i
	\end{matrix} \right]_2^N &= \left[ \begin{matrix}
	\rho' - r_4 & r_1 & \rho' - j\\
	\rho' - r_2 & r_3 & \rho' - i
	\end{matrix} \right]_2^N \\= \left[ \begin{matrix}\nonumber
	\rho'' - r_3 & r_2 & \rho'' - j\\
	\rho'' - r_1 & r_4 & \rho'' - i
	\end{matrix} \right]_2^N &= \left[ \begin{matrix}
	r_2 & \rho'' - r_3 & \rho'' - j\\
	r_4 & \rho'' - r_1 & \rho'' - i
	\end{matrix} \right]_2^N \\= \left[ \begin{matrix}\nonumber
	\rho-r_3 & \rho' - r_4 & \rho'' - j\\
	\rho-r_1 & \rho' - r_2 & \rho'' - i
	\end{matrix} \right]_2^N &= \left[ \begin{matrix}
	\rho' - r_4 & \rho-r_3 & \rho'' - j\\
	\rho' - r_2 & \rho-r_1 & \rho'' - i
	\end{matrix} \right]_2^N \\= \left[ \begin{matrix}\nonumber
	\rho''-r_3 & \rho-r_4 & \rho' - j\\
	\rho''-r_1 & \rho-r_2 & \rho' - i
	\end{matrix} \right]_2^N &= \left[ \begin{matrix}
	\rho-r_4 & \rho' - r_3 & \rho' - j\\
	\rho-r_2 & \rho' - r_1 & \rho' - i
	\end{matrix} \right]_2^N .
\end{align}
The structure of isomorphism $I\cong \mathbb{S}_3 \times \mathbb{Z}_2$ is as follows:
$$\begin{tabular}{|c|c|}
\hline 
()() & (12)() \\ 
\hline 
(12)(ab)& ()(ab) \\ 
\hline 
(13)(ab)& (132)(ab) \\ 
\hline 
(23)(ab)& (123)(ab) \\ 
\hline 
(123)() & (23)() \\ 
\hline 
(132)()& (13)() \\ 
\hline 
\end{tabular}  $$

The Regge transformation is the only new relation here:
\begin{empheq}[box=\fbox]{align}
\left[ \begin{matrix}
r_1 & r_2 & i\\
r_3 & r_4 & j
\end{matrix} \right]_2^N  =  \left[ \begin {matrix}\label{regge_II}
{ r_1}&\rho' - r_4&\rho' - j\\
{ r_3}&\rho' - r_2&\rho' - i
\end {matrix} \right]_2^N = \left[ \begin {matrix}
\rho'' - r_3 &r_2&\rho'' - j\\
\rho'' - r_1&r_4&\rho'' - i
\end {matrix} \right]_2^N 
\end{empheq}

Let us give a couple of examples of these symmetries.
\begin{itemize}
	\item Regge symmetry analogue, type II (1\textsuperscript{st} column is invariant, $N\ge2$):
	\begin{align*}
	\left\lbrace \begin{matrix}
	[5] & [6] & [10,1]\\
	\overline{[3]} & [8] & [7,5^{N-2}]
	\end{matrix} \right\rbrace=\left\lbrace \begin{matrix}
	[5] & {[7]} & [10,2]\\
	\overline{[3]} & [9] & [6,5^{N-2}]
	\end{matrix} \right\rbrace,\
	\left\lbrace \begin{matrix}
	[5] & {[6]} & [11]\\
	\overline{[1]} & [10] & [7,6^{N-2}]
	\end{matrix} \right\rbrace&=\left\lbrace \begin{matrix}
	[5] & {[7]} & [11,1]\\
	\overline{[1]} & [11] & [6,6^{N-2}]
	\end{matrix} \right\rbrace,\\
	\left\lbrace \begin{matrix}
	[4] & {[6]} & [10]\\
	\overline{[1]} & [9] & [7,6^{N-2}]
	\end{matrix} \right\rbrace=\left\lbrace \begin{matrix}
	[4] & {[7]} & [10,1]\\
	\overline{[1]} & [10] & [6,6^{N-2}]
	\end{matrix} \right\rbrace,\
	\left\lbrace \begin{matrix}
	[3] & {[6]} & [8,1]\\
	\overline{[4]} & [5] & [8,5^{N-2}]
	\end{matrix} \right\rbrace&=\left\lbrace \begin{matrix}
	[3] & {[8]} & [8,3]\\
	\overline{[4]} & [7] & [6,5^{N-2}]
	\end{matrix} \right\rbrace.
	\end{align*}
	\item Regge symmetry analogue, type II (2\textsuperscript{nd} column is invariant, $N\ge2$):
	\begin{align*}
	\left\lbrace \begin{matrix}
	[4] & [2] & [6]\\
	\overline{[1]} & [5] & [3,2^{N-2}]
	\end{matrix} \right\rbrace=\left\lbrace \begin{matrix}
	[6] & {[2]} & [6,2]\\
	\overline{[3]} & [5] & [1]
	\end{matrix} \right\rbrace,&\
	\left\lbrace \begin{matrix}
	[4] & {[3]} & [6,1]\\
	\overline{[1]} & [6] & [2,2^{N-2}]
	\end{matrix} \right\rbrace=\left\lbrace \begin{matrix}
	[5] & {[3]} & [6,2]\\
	\overline{[2]} & [6] & [1,1^{N-2}]
	\end{matrix} \right\rbrace,\\
	\left\lbrace \begin{matrix}
	[5] & {[6]} & [10,1]\\
	\overline{[4]} & [7] & [10,6^{N-2}]
	\end{matrix} \right\rbrace=\left\lbrace \begin{matrix}
	[10] & {[6]} & [10,6]\\
	\overline{[9]} & [7] & [5,1^{N-2}]
	\end{matrix} \right\rbrace,&\
	\left\lbrace \begin{matrix}
	[5] & {[6]} & [9,2]\\
	\overline{[2]} & [9] & [4,4^{N-2}]
	\end{matrix} \right\rbrace=\left\lbrace \begin{matrix}
	[7] & {[6]} & [9,4]\\
	\overline{[4]} & [9] & [2,2^{N-2}]
	\end{matrix} \right\rbrace.
	\end{align*}

\end{itemize}

\subsection{Type I external symmetries}\label{SS4}
In this subsection we consider external symmetries from the group $S$. 

\begin{nota}
	Let us denote by $\cong$ a external symmetry between two 6-j symbols with additional inequality restriction $M \ge 2$. For brevity we also drop out factors that occur in equalities and can be written as $C=(-1)^{N-M}\frac{K_T(r_1,r_2,r_3,r_4,i,j,N)}{K_T(\widetilde{r}_1,\widetilde{r}_2,\widetilde{r}_3,\widetilde{r}_4,\widetilde{i},\widetilde{j},M)}$.
\end{nota}

Let us consider external symmetries of type I. It is convenient to write down not the whole set $S\setminus I$, but the factor $E=S/I$. In $U_q(sl_2)$ we have the subgroups of Regge transformations, row and column permutations, one can notice that here we also have similar subgroups. The external symmetries for type I are analogous to column permutations and may be easily written with notations $\Delta_i=N-M_i$, $n_i=M_i-2$, $n_0=N-2$.

\begin{align}
\left[ \begin{matrix}
r_1 & r_2 & i\\
r_3 & r_4 & j
\end{matrix} \right]_1^N \cong &\left[ \begin{matrix}\label{w1s1}
r_1 & i+\Delta_1 & r_2+\Delta_1\\
r_3 & j+\Delta_1 & r_4+\Delta_1
\end{matrix} \right]_1^{M_1}  \cong \left[ \begin{matrix}
i+n_0 & r_2+\Delta_2 & r_1-n_2\\
j+n_0 & r_4+\Delta_2 & r_3-n_2
\end{matrix} \right]_1 ^{M_2}\\ \cong &\left[ \begin{matrix}\label{w1s2}
i+n_0 & r_1-n_3 & r_2+\Delta_3\\
j+n_0 & r_3-n_3 & r_4+\Delta_3
\end{matrix} \right]_1 ^{M_3}  \cong \left[ \begin{matrix}
r_2+n_0 & i+\Delta_4 & r_1-n_4 \\
r_4+n_0 & j+\Delta_4 & r_3-n_4
\end{matrix} \right]_1 ^{M_4}\\ 
\stackrel{N=M_5=2}{\cong} &\left[ \begin{matrix}\label{w1s3}
r_2+n_0 & r_1-n_5 & i+\Delta_5\\
r_4+n_0 & r_3-n_5 & j+\Delta_5
\end{matrix} \right]_1^{M_5},
\end{align}
where $n_i$, $\Delta_i$ and $M_i$ are fixed by fusion rules. 

We emphasize that $E$ is isomorphic to $\mathbb{S}_3$ only for $N=M_5=2$. In this case 6 elements from above are represented by $\{(),(23),(13),(132),(123),\textbf{(12)}\}$ correspondingly. In general, it is impossible to satisfy the fusion rules, so $E$ have only 4 transformations which are not closed under composition and $E \cong \mathbb{S}_3\setminus \{(12)\}$.

These symmetries are interesting because they cannot be expressed as a combination of any known symmetries. From hypergeometric point of view these symbols have the same value of ${}_4\Phi_3$ but it's still possible that $K_T$ is changed by this transformation.

Let us write down a few examples of these symmetries:
\begin{itemize}
	\item The first symmetry, $N=M_1=4$:
	\begin{align*}
	\left\lbrace \begin{matrix}
	[3] & \overline{[1]} & [4,1^{2}]\\
	[6] & [8] & [5,1^{2}]
	\end{matrix} \right\rbrace& = \left\lbrace \begin{matrix}
	[3] & \overline{[4]} & [1,1^{2}]\\
	[6] & [5] & [8,3^{2}]
	\end{matrix} \right\rbrace.
	\end{align*}
	\item The second symmetry, $N=4, M_2=3$:
	\begin{align*}
	\left\lbrace \begin{matrix}
	[5] & \overline{[4]} & [1]\\
	[7] & [8] & [9,3^{2}]
	\end{matrix} \right\rbrace& = -\sqrt{\dfrac{[2]_q[3]_q}{[5]_q[8]_q}}\left\lbrace \begin{matrix}
	[3] & \overline{[5]} & [4,3]\\
	[11] & [9] & [6]
	\end{matrix} \right\rbrace.
	\end{align*}
	\item The third symmetry, $N=4,M_3=2$:
	\begin{align*}
	\left\lbrace \begin{matrix}
	[7] & \overline{[3]} & [4]\\
	[2] & [6] & [1,1^{2}]
	\end{matrix} \right\rbrace& = \sqrt{\dfrac{[2]_q[3]_q}{[5]_q[6]_q}}\left\lbrace \begin{matrix}
	[6] & \overline{[7]} & [5]\\
	[3] & [2] & [8]
	\end{matrix} \right\rbrace.
	\end{align*}\item The fourth symmetry, $N=4,M_4=5$:
	\begin{align*}
	\left\lbrace \begin{matrix}
	[6] & \overline{[4]} & [8,3^{2}]\\
	[5] & [7] & [9,4^{2}]
	\end{matrix} \right\rbrace& = -\dfrac{[2]_q[3]_q}{[7]_q}\sqrt{\dfrac{1}{[6]_q^3}}\left\lbrace \begin{matrix}
	[6] & \overline{[7]} & [3,2^{3}]\\
	[9] & [8] & [2]
	\end{matrix} \right\rbrace.
	\end{align*}
\end{itemize}

\subsection{Type II external symmetries}\label{SS5}

In a similar way we can consider type II symmetries $E=S/I$ and fix $M$ by transformation (\ref{Nshift}). These symmetries are analogous to a column permutation and row permutations:
\begin{align}
\hspace{-5mm}&\left[ \begin{matrix}\label{w2s1}
r_1 & r_2 & i\\
r_3 & r_4 & j
\end{matrix} \right]_2^N \cong  \left[ \begin{matrix} 
 j + \Delta_1 & r_1 + \Delta_1 & r_4\\
 i & r_3 & r_2 + \Delta_1
\end{matrix} \right]_2^{M_1}  \cong  \left[ \begin{matrix}
r_2 + \Delta_2 & j + \Delta_2 & r_3\\
r_4 & i & r_1 + \Delta_2
\end{matrix} \right]_2^{M_2} \cong\\\hspace{-5mm}\cong &\left[ \begin{matrix}\nonumber
i-n_3 & r_1+\Delta_3 & r_2+n_0\\
j+n_0 & r_3 & r_4-n_3
\end{matrix} \right]_2^{M_3}  \cong  \left[ \begin{matrix}\nonumber
r_2+\Delta_4 & i-n_4 & r_1+n_0\\
r_4 & j+n_0 & r_3-n_4
\end{matrix} \right]_2^{M_4} \cong  \left[ \begin{matrix}\nonumber
r_4-n_5 & j+\Delta_5 & r_1+n_0\\
r_2+n_0 & i & r_3-n_5
\end{matrix} \right]_2^{M_5} \cong\\\hspace{-5mm}\cong  &\left[ \begin{matrix}\nonumber
j+\Delta_6 & r_3-n_6 & r_2+n_0\\
i & r_1+n_0 & r_4-n_6
\end{matrix} \right]_2^{M_6} \cong  \left[ \begin{matrix}\nonumber
r_1+\Delta_7 & r_4-n_7 & j+n_0\\
r_3 & r_2+n_0 & i-n_7
\end{matrix} \right]_2^{M_7} \cong  \left[ \begin{matrix}\nonumber
r_3-n_8 & r_2+\Delta_8 & j+n_0\\
r_1+n_0 & r_4 & i-n_8
\end{matrix} \right]_2^{M_8} \cong\\\hspace{-5mm}\cong  &\left[ \begin{matrix}\nonumber
r_4-n_9 & i-n_9 & r_3\\
r_2+n_0 & j+n_0 & r_1+\Delta_9
\end{matrix} \right]_2^{M_9} \cong  \left[ \begin{matrix}\nonumber
i-n_{10} & r_3-n_{10} & r_4\\
j+n_0 & r_1+n_0 & r_2+\Delta_{10}
\end{matrix} \right]_2^{M_{10}} \stackrel{N=M_{11}=2}{\cong}  \left[ \begin{matrix}
r_3-n_{11} & r_4-n_{11} & i\\
r_1+n_0 & r_2+n_0 & j+\Delta_{11}
\end{matrix} \right]_2^{M_{11}} \nonumber.
\end{align}

We emphasize that the last 6-j symbol exists only for $N=M_{11}=2$ as it is impossible to satisfy the inequalities otherwise. The isomorphism $E\big{|}_{N=M_{11}=2}\cong \mathbb{A}_4$ is as follows. The first row correspond to elements $\{(),(143), (134)\}$. The first column is presented by $\{(), (132), (234), (243)\}$. Other elements can be read from the table:
$$\begin{tabular}{|c|c|c|}
\hline 
() & (143) & (134) \\ 
\hline 
(132) & (123) & (142) \\ 
\hline 
(234) & (14)(23) & (13)(24) \\ 
\hline 
(243) & (124) & \textbf{(12)(34)} \\ 
\hline 
\end{tabular} $$

If we consider arbitrary $N$, $E$ is not closed under compositions and $E \cong \mathbb{A}_4 \setminus \{(12)(34)\}$.

Let us write down a few examples of these symmetries:
\begin{itemize}
	\item The first symmetry, $N=M_1=4$:
	\begin{align*}
	\left\lbrace \begin{matrix}
	[5] & [2] & [7]\\
	\overline{[4]} & [3] & [6,2^{2}]
	\end{matrix} \right\rbrace& = \left\lbrace \begin{matrix}
	[6] & [5] & [7,4]\\
	\overline{[7]} & [4] & [2,2^{2}]
	\end{matrix} \right\rbrace.
	\end{align*}
	\item The third symmetry, $N=4, M_3=5$:
	\begin{align*}
	\left\lbrace \begin{matrix}
	[8] & [1] & [9]\\
	\overline{[5]} & [4] & [6,1^{2}]
	\end{matrix} \right\rbrace& = -\sqrt{\dfrac{[10]_q[8]_q}{[4]_q[3]_q}}\left\lbrace \begin{matrix}
	[6] & [7] & [8,5]\\
	\overline{[8]} & [5] & [1,1^{3}]
	\end{matrix} \right\rbrace.
	\end{align*}
	\item The seventh symmetry, $N=4,M_7=5$:
	\begin{align*}
	\left\lbrace \begin{matrix}
	[8] & [1] & [9]\\
	\overline{[5]} & [4] & [6,1^{2}]
	\end{matrix} \right\rbrace& = -\sqrt{\dfrac{[10]_q[8]_q}{[4]_q[3]_q}}\left\lbrace \begin{matrix}
	[7] & [1] & [8]\\
	\overline{[5]} & [3] & [6,1^{3}]
	\end{matrix} \right\rbrace.
	\end{align*}\item The ninth symmetry, $N=4,M_9=3$:
	\begin{align*}
	\left\lbrace \begin{matrix}
	[8] & [4] & [6,3]\\
	\overline{[3]} & [9] & [5,3^{2}]
	\end{matrix} \right\rbrace& = -\sqrt{\dfrac{[2]_q^2[3]_q[6]_q^2[10]_q}{[4]_q^3[5]_q[11]_q[12]_q}}\left\lbrace \begin{matrix}
	[8] & [5] & [8,5]\\
	\overline{[6]} & [7] & [9,5]
	\end{matrix} \right\rbrace.
	\end{align*}
\end{itemize}

\section{Main results}

In this section we collect the most important results obtained in the paper. We are using the special notation (\ref{MFS_nota_t1},\ref{MFS_nota_t2}) for MFS.
\begin{itemize}
 \item Expression (\ref{final_expr_hyp}) for MFS via q-hypergeometric series:
\begin{gather}
\left[ \begin{matrix}
r_1 & r_2 & i\\
r_3 & r_4 & j
\end{matrix} \right]_T^N
=K_T\cdot {}_4\Phi_3 \left( \begin{matrix}
a_1,a_2,a_3,a_4\\ b_1,b_2,b_3
\end{matrix}; q,q \right),\\
2 a_i = \left( \begin{matrix}
-r_1-r_2+i - 2(N-2)\delta_{T,2}\\ -r_3-r_4+i\\ -r_1-r_4+j\\ -r_2-r_3+j
\end{matrix} \right),
\qquad 2 b_i = \left( \begin{matrix}  -r_1-r_2-r_3-r_4-2(N-1)\\ i+j-r_2-r_4+2\\  i+j-r_1-r_3 + 2 + 2(N-2)\delta_{T,1}\\\end{matrix} \right).
\end{gather}
Factor $K_T$ depends on type $T$ and defined as in (\ref{coef_K}):
{\small\begin{equation}\hspace{-10mm}
	K_T = \dfrac{\theta_N\left(r_1,r_2,i\right) \theta_N\left(r_3, r_4, i\right) \theta_N\left(r_1, r_4, j\right) \theta_N\left(r_2, r_3, j\right) [N-1]_q![N-2]_q! [\frac{r_1+r_2+r_3+r_4}{2}+N-1]_q! }{[\frac{r_3+r_4-i}{2}]_q! [\frac{r_1+r_2-i}{2} + (N-2)\delta_{T,2}]_q! [\frac{r_2+r_3-j}{2}]_q! [\frac{r_1+r_4-j}{2}]_q! [\frac{i+j-r_2-r_4}{2}]_q! [\frac{i+j-r_1-r_3}{2}+(N-2)\delta_{T,1}]_q!}.
	\end{equation}}
\item Relation (\ref{rollback}) between MFS and $U_q(sl_2)$ 6-j symbols:
\begin{align}
\left[ \begin{matrix}
r_1 & r_2 & i\\
r_3 & r_4 & j
\end{matrix} \right]_1^N
&=\left\{\begin{matrix}
r_1& r_2+N-2 &i+N-2 \\
r_3 & r_4+N-2  & j+N-2
\end{matrix} \right\}(-1)^N [N-1]_q! [N-2]_q! \cdot \Theta_1(N),\\
\left[ \begin{matrix}
r_1 & r_2 & i\\
r_3 & r_4 & j
\end{matrix} \right]_2^N
&=\left\{\begin{matrix}
r_1+N-2 & r_2+N-2 &i\\
r_3 & r_4 & j+N-2
\end{matrix} \right\}(-1)^N [N-1]_q! [N-2]_q! \cdot \Theta_2(N),
\end{align}
with factors $\Theta_1,\Theta_2$ defined in (\ref{coef_theta}):
{\small\begin{align}
\Theta_1(N)&=
\left(\prod_{m=1}^{N-2} \left[\frac{i {-} r_1 {+} r_2}{2}+m\right]_q \left[\frac{j {+} r_2 {-} r_3}{2}+m\right]_q  \left[\frac{j {-} r_1 {+} r_4}{2}+m\right]_q \left[\frac{i {-} r_3 {+} r_4}{2}+m\right]_q \right)^{-\frac{1}{2}},\\
\Theta_2(N)&=
\left(\prod_{m=1}^{N-2} \left[\frac{r_1 {+} r_2{-}i}{2}+m\right]_q \left[\frac{j {+} r_2 {-} r_3}{2}+m\right]_q  \left[\frac{j {+} r_1 {-} r_4}{2}+m\right]_q \left[\frac{i {+} r_3 {+} r_4}{2}+1+m\right]_q \right)^{-\frac{1}{2}}.
\end{align}}
\item The asymptotics (\ref{asympt}) of MFS for $U(sl_N)$:
\begin{equation}
\dfrac{1}{\Theta_T(N)}\left[ \begin{matrix}
r_1 & r_2 & i\\
r_3 & r_4 & j
\end{matrix} \right]_T^N
\sim \dfrac{(-1)^N\cdot(N-1)!\cdot(N-2)!}{\sqrt{12\pi\cdot |V(\widetilde{J}_k)|}} \cos\left(\sum_{n=1}^{6} \widetilde{J}_k \cdot \Omega(\widetilde{J}_k) + \dfrac{\pi}{4} \right),
\end{equation}
where $\widetilde{J}_k$ are defined in (\ref{J_def}).
\end{itemize}
\subsection{Symmetries of 6-j symbols}
There is a group of MFS symmetries that has 144 elements in total. It is convenient to split them into the internal and external symmetries. The internal ones always act in $U_q(sl_N)$, the external ones connect $U_q(sl_N)$ and $U_q(sl_M)$ 6-j symbols.

\begin{itemize}
	\item Counterpart of the Regge transformations (\ref{regge_I}), type I ($\rho' = \frac{r_1+r_3+i +j}{2} = \frac{r_2+r_4+i +j}{2}$):
	\begin{align}
	\left[ \begin{matrix}
	r_1 & r_2 & i\\
	r_3 & r_4 & j
	\end{matrix} \right]_1^N =& \left[ \begin {matrix} 
	{ r_1}&\rho' - j&\rho' - r_4\\ { r_3}&\rho' - i&\rho' - r_2\end {matrix} \right]_1^N=
	\left[ \begin {matrix}
	\rho' - j &r_2&\rho' - r_3\\ \rho' - i&r_4&\rho' - r_1\end {matrix} \right]_1^N.
	\end{align}
\item Counterpart of the Regge transformations (\ref{regge_II}), type II ($\rho' = \frac{r_1+r_3+i +j}{2}, \rho''= \frac{r_2+r_4+i +j}{2}$):
\begin{align}
\left[ \begin{matrix}
r_1 & r_2 & i\\
r_3 & r_4 & j
\end{matrix} \right]_2^N  =  \left[ \begin {matrix} { r_1}&\rho' - r_4&\rho' - j\\ { r_3}&\rho' - r_2&\rho' - i\end {matrix} \right]_2^N =
\left[ \begin {matrix}  \rho'' - r_3 &r_2&\rho'' - j\\ \rho'' - r_1&r_4&\rho'' - i\end {matrix} \right]_2^N.
\end{align}
\end{itemize}

The next symmetries are between $U_q(sl_N)$ and $U_q(sl_M)$ 6-j symbols. Values of $M_i$ are fixed by fusion rules. For brevity we use the notations $\Delta_i=N-M_i$, $n_i=M_i-2$, $n_0=N-2$.
\begin{itemize}
	\item Type I external symmetries (\ref{w1s1}):

{\small\begin{align}
&\left[ \begin{matrix}
r_1 & r_2 & i\\
r_3 & r_4 & j
\end{matrix} \right]_1^N \cong \left[ \begin{matrix}
r_1 & i+\Delta_1 & r_2+\Delta_1\\
r_3 & j+\Delta_1 & r_4+\Delta_1
\end{matrix} \right]_1^{M_1}  \cong \left[ \begin{matrix}
i+n_0 & r_2+\Delta_2 & r_1-n_2\\
j+n_0 & r_4+\Delta_2 & r_3-n_2
\end{matrix} \right]_1 ^{M_2}\cong\\ \cong &\left[ \begin{matrix}\nonumber
i+n_0 & r_1-n_3 & r_2+\Delta_3\\
j+n_0 & r_3-n_3 & r_4+\Delta_3
\end{matrix} \right]_1 ^{M_3}  \cong \left[ \begin{matrix}
r_2+n_0 & i+\Delta_4 & r_1-n_4 \\
r_4+n_0 & j+\Delta_4 & r_3-n_4
\end{matrix} \right]_1 ^{M_4} \stackrel{N=M_5=2}{\cong}  \left[ \begin{matrix}
r_2+n_0 & r_1-n_5 & i+\Delta_5\\
r_4+n_0 & r_3-n_5 & j+\Delta_5
\end{matrix} \right]_1^{M_5}.
\end{align}}

\item Type II external symmetries (\ref{w2s1}):

{\small\begin{align}
	\hspace{-5mm}&\left[ \begin{matrix}
	r_1 & r_2 & i\\
	r_3 & r_4 & j
	\end{matrix} \right]_2^N \cong  \left[ \begin{matrix} 
	j + \Delta_1 & r_1 + \Delta_1 & r_4\\
	i & r_3 & r_2 + \Delta_1
	\end{matrix} \right]_2^{M_1}  \cong  \left[ \begin{matrix}
	r_2 + \Delta_2 & j + \Delta_2 & r_3\\
	r_4 & i & r_1 + \Delta_2
	\end{matrix} \right]_2^{M_2} \cong\\\hspace{-5mm}\cong &\left[ \begin{matrix}\nonumber
	i-n_3 & r_1+\Delta_3 & r_2+n_0\\
	j+n_0 & r_3 & r_4-n_3
	\end{matrix} \right]_2^{M_3}  \cong  \left[ \begin{matrix}\nonumber
	r_2+\Delta_4 & i-n_4 & r_1+n_0\\
	r_4 & j+n_0 & r_3-n_4
	\end{matrix} \right]_2^{M_4} \cong  \left[ \begin{matrix}\nonumber
	r_4-n_5 & j+\Delta_5 & r_1+n_0\\
	r_2+n_0 & i & r_3-n_5
	\end{matrix} \right]_2^{M_5} \cong\\\hspace{-5mm}\cong  &\left[ \begin{matrix}\nonumber
	j+\Delta_6 & r_3-n_6 & r_2+n_0\\
	i & r_1+n_0 & r_4-n_6
	\end{matrix} \right]_2^{M_6} \cong  \left[ \begin{matrix}\nonumber
	r_1+\Delta_7 & r_4-n_7 & j+n_0\\
	r_3 & r_2+n_0 & i-n_7
	\end{matrix} \right]_2^{M_7} \cong  \left[ \begin{matrix}\nonumber
	r_3-n_8 & r_2+\Delta_8 & j+n_0\\
	r_1+n_0 & r_4 & i-n_8
	\end{matrix} \right]_2^{M_8} \cong\\\hspace{-5mm}\cong  &\left[ \begin{matrix}\nonumber
	r_4-n_9 & i-n_9 & r_3\\
	r_2+n_0 & j+n_0 & r_1+\Delta_9
	\end{matrix} \right]_2^{M_9} \cong  \left[ \begin{matrix}\nonumber
	i-n_{10} & r_3-n_{10} & r_4\\
	j+n_0 & r_1+n_0 & r_2+\Delta_{10}
	\end{matrix} \right]_2^{M_{10}} \stackrel{N=M_{11}=2}{\cong}  \left[ \begin{matrix}
	r_3-n_{11} & r_4-n_{11} & i\\
	r_1+n_0 & r_2+n_0 & j+\Delta_{11}
	\end{matrix} \right]_2^{M_{11}} \nonumber.
	\end{align}}

\end{itemize}

\section{Conclusion}
The 6-j symbols beyond $U_q(sl_2)$ are rapidly becoming very complicated to analyze. Even in the case of symmetric and conjugate to symmetric representations where we know the analytic expression, there are many features that hide from our sight. Firstly, 6-j expression in its original form \cite{MFS} is the q-factorial series that can be written as a function $_5\Phi_4$, but after some manipulations it became clear that the expression is very similar to $U_q(sl_2)$ one and may be written as (\ref{final_expr}) via $_4\Phi_3$.

Secondly, the hypergeometric function has a relation (\ref{hyp_relation}) that is necessary to use the Sears' transformation. This allow us to think that there is an important class of 6-j symbols with 6 free parameters that is connected with $_4\Phi_3$ series. Considered expression (\ref{final_expr}) is already applicable to $N=2$ case and types I, II. It is an interesting question what else may be expressed via $_4\Phi_3$.

The relation (\ref{rollback}) between multiplicity-free $U_q(sl_N)$ and $U_q(sl_2)$ symbols reveals the nature of multiplicity-free case. In fact, multiplicity-free 6-j symbols tends to be very similar to $U_q(sl_2)$ one. As was found in \cite{3SB, Alekseev:2019}, the other class of 6-j symbol with symmetric incoming representations may be expressed via $U_q(sl_2)$ one. The further study of more difficult classes can tell us more about the structure of 6-j symbols, but now we can vividly see that q-hypergeometric series play the main role in this problem.

Obtained symmetries show that there are much more relations for 6-j symbols in $U_q(sl_N)$ than tetrahedral symmetries. As the most bright example of this statement, we show that the Regge symmetry is generalizable to both types as (\ref{regge_I}, \ref{regge_II}). External symmetries, on the other hand, are less convenient to use, but they provide a lot of new relations that depend on $N$ explicitly and connects 6-j symbols from different $N$.

\section*{Appendix}

In this Appendix we write down explicitly the new symmetries mentioned in select results. Also we prove that external symmetries always preserve the non-negativeness of $\widetilde{r}_n,\widetilde{i},\widetilde{j}$.

\begin{itemize}
	\item Counterpart of the Regge transformations (\ref{regge_I}) in terms of Young diagrams, type I:
	{\small\begin{align}
		\left\lbrace \begin{matrix}
		[r_1] & \overline{[r_2]} & \left[ i, \frac{r_2-r_1+i}{2}^{N-2} \right]\\
		[r_3] & [r_4] & \left[ j, \frac{r_2-r_3+j}{2}^{N-2} \right]
		\end{matrix} \right\rbrace =& \left\lbrace \begin {matrix}
		[r_1]& \overline{\left[\dfrac{r_2+r_4-i+j}{2}\right]} &\left[\dfrac{-r_2+r_4+i+j}{2},\dfrac{r_3-r_2+j}{2}^{N-2}\right]\\  [r_3]&\left[\dfrac{r_2+r_4+i-j}{2}\right]&\left[\dfrac{r_2-r_4+i+j}{2}, \dfrac{r_2-r_3+j}{2}^{N-2}\right]
		\end {matrix} \right\rbrace=\\
		=&\left\lbrace \begin {matrix}
		\left[\dfrac{r_1+r_3-i+j}{2}\right] &\overline{[r_2]}&\left[\dfrac{-r_1+r_3+i+j}{2} , \dfrac{r_2-r_1+i}{2}^{N-2} \right]\\
		\left[\dfrac{r_1+r_3+i-j}{2}\right]& [r_4] &\left[\dfrac{r_1-r_3+i+j}{2}, \dfrac{r_2-r_3+j}{2}^{N-2}\right]
		\end {matrix} \right\rbrace\nonumber.
		\end{align}}
	\item Counterpart of the Regge transformations (\ref{regge_II})  in terms of Young diagrams, type II:
{\small\begin{align}
	\left\lbrace \begin{matrix}
	[r_1] & [r_2] & \left[\frac{r_1+r_2+i}{2}, \frac{r_1+r_2-i}{2} \right]\\
	\overline{[r_3]} & [r_4] & \left[ j, \frac{r_2-r_3+j}{2}^{N-2} \right]
	\end{matrix} \right\rbrace  = & \left\lbrace \begin {matrix}
	[r_1]& \left[\dfrac{r_2-r_4+i+j}{2}\right] &\left[\dfrac{r_1+r_2+i}{2},\dfrac{r_1-r_4+j}{2}\right]\\  \overline{[r_3]}&\left[\dfrac{r_4-r_2+i+j}{2}\right]&\left[\dfrac{r_2+r_4-i+j}{2}, \dfrac{r_2-r_3+j}{2} ^{N-2}\right]
	\end {matrix} \right\rbrace = \\
	= &\left\lbrace \begin {matrix}
	\left[\dfrac{r_1-r_3+i+j}{2}\right] &[r_2]&\left[\dfrac{r_1+r_2+i}{2} , \dfrac{r_2-r_3+j}{2} \right]\\
	\overline{\left[\dfrac{r_3-r_1+i+j}{2}\right]}& [r_4] &\left[\dfrac{r_1+r_3-i+j}{2}, \dfrac{r_2+r_1-i}{2} ^{N-2}\right]
	\end {matrix} \right\rbrace\nonumber.
	\end{align}}

	\item Type I external symmetries (\ref{w1s1}):

	\begin{align}
	\left[ \begin{matrix}
	r_1 & r_2 & i\\
	r_3 & r_4 & j
	\end{matrix} \right]_1^N \cong &\left[ \begin{matrix}
	r_1 & \dfrac{r_2+r_4+i-j}{2} & \dfrac{3r_2+r_4-i-j}{2}\\
	r_3 & \dfrac{r_2+r_4-i+j}{2} & \dfrac{r_2+3r_4-i-j}{2}
	\end{matrix} \right]_1^{N+\frac{i+j-r_2-r_4}{2}} \\\cong &\left[ \begin{matrix}
	i+N-2 & \dfrac{i+j+r_2-r_4}{2} +N-2& \dfrac{i+j+r_1-r_3}{2}\\
	j+N-2 & \dfrac{i+j-r_2+r_4}{2}+N-2 & \dfrac{i+j-r_1+r_3}{2}
	\end{matrix} \right]_1^{2+\frac{r_2+r_4-i-j}{2}}\\ \cong &\left[ \begin{matrix}
	i+N-2 & \dfrac{i+j+r_1-r_3}{2} & \dfrac{i+j+r_2-r_4}{2}+N-2\\
	j+N-2 & \dfrac{i+j-r_1+r_3}{2} & \dfrac{i+j-r_2+r_4}{2}+N-2
	\end{matrix} \right]_1 ^{2+\frac{r_1+r_3-i-j}{2}} \\ \cong&\left[ \begin{matrix}
	r_2+N-2 & \dfrac{r_2+r_4+i-j}{2} +N-2& \dfrac{3r_1+r_3-i-j}{2}\\
	r_4+N-2 & \dfrac{r_2+r_4-i+j}{2}+N-2 & \dfrac{r_1+3r_3-i-j}{2}
	\end{matrix} \right]_1 ^{2+\frac{i+j-r_2-r_4}{2}}\\ \cong &\left[ \begin{matrix}
	r_2+N-2 & r_1+N-2 & i+2N-4\\
	r_4+N-2 & r_3+N-2 & j+2N-4
	\end{matrix} \right]_1^{4-N}.
	\end{align}

	\item Type II external symmetries (\ref{w2s1}):

{\small\begin{align}
	\left[ \begin{matrix}
	r_1 & r_2 & i\\
	r_3 & r_4 & j
	\end{matrix} \right]_2^N \cong &\left[ \begin{matrix}
	\frac{r_3+i-r_1+j}{2} & \frac{r_3+i+r_1-j}{2} & r_4\\
	i & r_3 & \frac{3r_2+i-r_4-j}{2}
	\end{matrix} \right]_2^{N+\frac{r_1+j-r_3-i}{2}}  \\\cong & \left[ \begin{matrix}
	\frac{r_4+i+r_2-j}{2} & \frac{r_4+i-r_2+j}{2} & r_3\\
	r_4 & i & \frac{3r_1+i-r_3-j}{2}
	\end{matrix} \right]_2^{N+\frac{r_2+j-r_4-i}{2}} \\\cong &\left[ \begin{matrix}
	\frac{r_3+j-r_1+i}{2} & \frac{r_3+j+r_1-i}{2}+N-2 & r_2+N-2\\
	j+N-2 & r_3 & \frac{r_4+j+r_2-i}{2}
	\end{matrix} \right]_2^{2+\frac{r_1+i-r_3-j}{2}} \\ \cong &\left[ \begin{matrix}
	\frac{r_4+j+r_2-i}{2}+N-2 & \frac{r_4+j-r_2+i}{2} & r_1+N-2\\
	r_4 & j+N-2 & \frac{r_3+j+r_1-i}{2}
	\end{matrix} \right]_2^{2+\frac{r_2+i-r_4-j}{2}} \\\cong &\left[ \begin{matrix}
	\frac{r_2+i+r_4-j}{2} & \frac{r_2+i-r_4+j}{2}+N-2 & r_1+N-2\\
	r_2+N-2 & i & \frac{3r_3+i-r_1-j}{2}
	\end{matrix} \right]_2^{2+\frac{r_4+j-r_2-i}{2}} \\\cong & \left[ \begin{matrix}
	\frac{r_1+i-r_3+j}{2}+N-2 & \frac{r_1+i+r_3-j}{2} & r_2+N-2\\
	i & r_1+N-2 & \frac{3r_4+i-r_2-j}{2}
	\end{matrix} \right]_2^{2+\frac{r_3+j-r_1-i}{2}} \\\cong & \left[ \begin{matrix}
	r_3+N-2 & r_2 & j+N-2\\
	r_3 & r_2+N-2 & i+r_3-r_1
	\end{matrix} \right]_2^{2+r_1-r_3} \\\cong & \left[ \begin{matrix}
	r_1 & r_4+N-2 & j+N-2\\
	r_1+N-2 & r_4 & i+r_1-r_3
	\end{matrix} \right]_2^{2+r_1-r_3} \\\cong &\left[ \begin{matrix}
	\frac{r_2+j+r_4-i}{2}+N-2 & \frac{r_2+j-r_4+i}{2}+4-N & r_3\\
	r_2+N-2 & j+N-2 & \frac{r_1+j+r_3-i}{2}+2N-4
	\end{matrix} \right]_2^{4-N+\frac{r_4+i-r_2-j}{2}} \\\cong & \left[ \begin{matrix}
	\frac{r_1+j-r_3+i}{2}+2-N & \frac{r_1+j+r_3-i}{2}+2-N & r_4\\
	j+N-2 & r_1+N-2 & \frac{r_2+j+r_4-i}{2}+2N-4
	\end{matrix} \right]_2^{4-N+\frac{r_3+i-r_1-j}{2}} \\\cong &\left[ \begin{matrix}
	r_3+N-2 & r_4+N-2 & i\\
	r_1+N-2 & r_2+N-2 & j+2N-4
	\end{matrix} \right]_2^{4-N}.
	\end{align}}

\end{itemize}

\subsection*{Proof of statement \ref{st5}}
\begin{proof}
	Let us firstly prove that the following expressions are non-negative:
	\begin{align}
	\begin{cases}
	r_1+r_3+i-j \ge 0, \\
	3r_2+r_4-i-j \ge 0,
	\end{cases} \qquad &T\in\{1,2\},\\
	i+r_1-r_3 \ge 0, \quad\qquad &T=2.
	\end{align}
	The non-negativity can be proven using the inequalities on $i,j$. These inequalities are tautological generalization of the $U_q(sl_2)$ case \cite{Alekseev:2019}:
	\begin{align}
	\max\begin{pmatrix}
	|r_1-r_2|\\
	|r_3-r_4|
	\end{pmatrix} \le i \le \min\begin{pmatrix}
	r_1+r_2\\
	r_3+r_4
	\end{pmatrix}, \qquad \max\begin{pmatrix}
	|r_2-r_3|\\
	|r_1-r_4|
	\end{pmatrix} \le j \le \min\begin{pmatrix}
	r_2+r_3\\
	r_1+r_4
	\end{pmatrix}.
	\end{align}
	With the suitable substitution the proof is obvious:
	\begin{equation}
	r_1+r_3+i-j \ge r_1+r_3+(-r_3+r_4)-(r_1+r_4) \ge 0,
	\end{equation}
	\begin{equation}
	3r_2+r_4-i-j \ge 3r_2+r_4-(r_1+r_2)-(r_2+r_3) \ge r_2+r_4-r_1-r_3 = 0,
	\end{equation}
	\begin{equation}
	i+r_1-r_3 \ge \max\begin{pmatrix}
	r_2-r_1+r_1-r_3\\
	r_3-r_4+r_1-r_3
	\end{pmatrix} =\max\begin{pmatrix}
	r_2-r_3\\
	r_1-r_4
	\end{pmatrix} =\max\begin{pmatrix}
	r_2-r_3\\
	r_3-r_2
	\end{pmatrix} \ge 0.
	\end{equation}
	Similarly one can derive non-negativeness of all expressions from external symmetries. Since the derivation is analogous in these cases, they are omitted. 
\end{proof}

\section*{Acknowledgments}
We are deeply indebted to Andrei Mironov and Alexei Morozov for numerous stimulating discussions. V.A. is also grateful to Satoshi Nawata for clarifications on fusion rules, to Andrei Zotov and Victor Mishnyakov for useful discussions and comments.

Our work was partly supported by the grant of the Foundation for the Advancement of Theoretical Physics ``BASIS" (A.M., A.S. and A.V.), by RFBR grants 19-01-00680 (V.A.), 17-01-00585 (A.M.), 18-31-20046 (A.S.), by joint RFBR grants 19-51-18006 (A.M.), 19-51-50008-Yaf-a (A.M.), 18-51-05015-Arm-a (A.M, A.S.), 18-51-45010-Ind-a (A.M, A.S.), 19-51-53014-GFEN-a (A.M, A.S.), 19-51-18006-Bolg-a (A.M.), by President of Russian Federation grant MK-2038.2019.1 (A.M.).

\bibliographystyle{unsrturl}
\bibliography{Hyp_geom_syms_9}{}

\end{document}